\RequirePackage[2020-02-02]{latexrelease}
\documentclass{rQUF2e}

\graphicspath{{./images/}}

\usepackage{epstopdf}
\usepackage{subfigure}
\usepackage{amsmath,amssymb,amsthm,bbm,enumerate,mathrsfs,dsfont,amsfonts,hyperref,float}

\usepackage{todonotes}

\usepackage{xcolor}
\pagecolor{white}
\usepackage{mathrsfs}
\theoremstyle{plain}
\newtheorem{theorem}{Theorem}
\newtheorem{corollary}{Corollary}
\newtheorem{lemma}{Lemma}
\newtheorem{proposition}{Proposition}
\newtheorem{example}{Example}

\theoremstyle{definition}
\newtheorem{definition}{Definition}
\newtheorem{asmp}{Assumption}

\theoremstyle{remark}
\newtheorem{remark}{Remark}

\usepackage[bbgreekl]{mathbbol}
\DeclareSymbolFontAlphabet{\mathbb}{AMSb}
\DeclareSymbolFontAlphabet{\mathbbl}{bbold}

\usepackage{natbib}
\title{\textit{Risk contributions of lambda quantiles}
\footnote{\textit{History:} Earlier versions have been presented at Birkbeck, University of London PhD Jamboree, the \textit{Mathematical and Statistical Methods for Actuarial Sciences and Finance  2020 Conference (eMAF)} (virtual), and the \textit{10th General Advanced Mathematical Methods for Finance (AMaMeF) 2021 Conference} (virtual).}
}

\author{A. INCE$^{\ast, \nshortmid}$\thanks{$^\nshortmid$Corresponding author.
Email: aince02@student.bbk.ac.uk}, I. PERI$^{\ast, \dag}$\thanks{$\dag$Email: i.peri@bbk.ac.uk} and S. PESENTI$^{\ddag}$\thanks{$\ddag$Email: silvana.pesenti@utoronto.ca}\\
\affil{$^\ast$Department of Economics, Mathematics and Statistics, Birkbeck, University of London, Malet Street, Bloomsbury, London WC1E 7HX, UK\\
$^\ddag$Department of Statistical Sciences, University of Toronto, 700 University Avenue, Toronto, ON M5G 1Z5, Canada} \received{March 30, 2022} }

\begin{document}

\maketitle

\begin{abstract}
Risk contributions of portfolios form an indispensable part of risk adjusted performance measurement. The risk contribution of a portfolio, e.g., in the Euler or Aumann-Shapley framework, is given by the partial derivatives of a risk measure applied to the portfolio profit and loss in direction of the asset units. For risk measures that are not positively homogeneous of degree 1, however, known capital allocation principles do not apply. We study the class of lambda quantile risk measures that includes the well-known Value-at-Risk as a special case but for which no known allocation rule is applicable. We prove differentiability and derive explicit formulae of the derivatives of lambda quantiles with respect to their portfolio composition, that is their risk contribution. For this purpose, we define lambda quantiles on the space of portfolio compositions and consider generic (also non-linear) portfolio operators. 

We further derive the Euler decomposition of lambda quantiles for generic portfolios and show that lambda quantiles are homogeneous in the space of portfolio compositions, with a homogeneity degree that depends on the portfolio composition and the lambda function. This result is in stark contrast to the positive homogeneity properties of risk measures defined on the space of random variables which admit a constant homogeneity degree. We introduce a generalised version of Euler contributions and Euler allocation rule, which are compatible with risk measures of any homogeneity degree and non-linear but homogeneous portfolios.
These concepts are illustrated by a non-linear portfolio  using financial market data.
\end{abstract}

\begin{keywords}
Lambda Quantiles; Capital Allocation; Risk Contribution; Lambda Value-at-Risk; Euler Allocation
\end{keywords}

\section{Introduction}

Calculating firm-wide or portfolio-level risk is at the heart of modern financial risk management. Financial institutions  use risk measures to determine economic capital, that is a capital buffer to absorb unexpected losses during adverse market scenarios and to preserve solvency. However, understanding how firm-wide or portfolio-level risk are formed and affected by their respective constituents is of equal importance to risk management processes. Determining contributions of assets or sub-portfolios to the overall portfolio risk, or contributions of business lines to the firm-wide risk enables market practitioners to make informed decisions on capital allocations to protect each business line's profitability and secure its solvency.

In this paper we focus on lambda quantile risk measures, a class of law-invariant risk measures that generalises the well-known risk measure Value-at-Risk (VaR). Lambda quantiles were first proposed by \cite{fmp14} to overcome two of the most criticised aspects of VaR. First, VaR's inability to distinguish between different tail behaviours and second its failure to capture extreme losses. Indeed, lambda quantiles have the ability to (a) penalise heavy-tailed (portfolio) distributions and (b) identify extreme losses dynamically, e.g., by recalibrating the lambda function of lambda quantiles, see \cite{hmp18}. The key difference between VaR and a lambda quantile is, that while $VaR_{\lambda}$ is the negative of a quantile function at fixed level $\lambda$, a lambda quantile is the negative of a generalised quantile at a level determined by a function -- the so-called lambda function. 

Throughout, we consider generic, not necessarily linear, portfolio operators, that are collections of linear and/or non-linear assets where both long and short positions are permitted. Here, we consider  portfolios consisting of a random vector of asset profits and losses and a portfolio composition, a vector containing the number of units of each asset. To calculate risk contributions, we define lambda quantiles on the space of portfolio compositions, a subset of $\mathbb{R}^n$, instead of the portfolio profit and loss, the space of random variables. Using this novel domain for lambda quantiles, we study how the portfolio's risk -- the lambda quantile of the portfolio -- is affected by changes in its composition. Specifically, we address the question of what each asset's contribution is to the overall portfolio risk. These risk contributions quantify the extent of change in portfolio risk due to changes in an asset's exposure; an important metric in portfolio rebalancing.

Lambda quantiles are the subject of extensive study in \cite{bpr17}, \cite{hmp18}, and \cite{cp18}, where lambda quantiles are referred to as Lambda Value-at-Risk. When defined on the set of probability measures, lambda quantiles possess the properties of monotonicity and quasi-convexity \citep{fmp14}. \cite{bpr17} study robustness, elicitability, and consistency properties of lambda quantiles. A theoretical framework for backtesting lambda quantiles is provided in \cite{cp18}, who propose three backtesting methodologies. Moreover, \cite{hmp18} argue to estimate the lambda function of lambda quantiles using major stock market indices, such as S\&P500, FTSE100, and EURO STOXX 50, which provides a dynamic macro approach to measuring market risk. The axiomatisation and further properties of lambda quantiles are studied in \cite{bp19}. These previous studies on lambda quantiles have either defined lambda quantiles on the space of probability measures or on the space of almost surely finite random variables. For the purpose of risk contributions, however, we define lambda quantiles on subsets of $\mathbb{R}^n$; the domain of asset compositions of a portfolio. Defining lambda quantiles on the space of portfolio compositions provides a natural way of comparing rates of change in portfolio risk with respect to asset units. Understanding changes in portfolio risk that may arise from portfolio rebalancing is highly important from a performance measurement perspective and relevant for risk capital allocation.

There exist a plethora of risk capital allocation methods that firms use for risk management and performance measurement, see \cite{b17} for a review and comparison of risk capital allocation methods and their properties. It should be noted that not all capital allocation methods are compatible with a specific risk measure, and applicability is determined by the properties of the risk measure in question. For example, the axiomatic approach taken in \cite{d01} to define a coherent risk capital allocation principle derived from the Aumann-Shapley value \citep{as74} only applies to coherent risk measures \citep{a99}; a property lambda quantiles do not possess. Furthermore, the Aumann-Shapley capital allocation rule introduced in \cite{ts09}, which was also inspired from the Aumann-Shapley value, is defined for Gateaux differentiable risk measures on linear portfolios. Explicit formulae of the Aumann-Shapley allocation rule is provided, for the class of convex risk measures, in \cite{ts09}. For positive homogeneous (but not necessarily coherent) risk measures, the Euler capital allocation \citep{p99,t99,d01,t07} can be used, which, on the space of coherent risk measures, coincides with the Aumann-Shapley allocation. It is worth noting that both \cite{d01} and \cite{t99} arrive (for coherent risk measures) at the same capital allocation rule using different theoretical approaches: the former uses a game-theoretic approach whilst the latter the notion of risk-adjusted performance measurement, a common practice for company internal economic capital calculations. \cite{kalkbrener2005MF} provides an axiomatic approach for sub-additive and homogeneous risk measures. The general class of non-Gateaux differentiable but convex or quasi-convex risk measures is treated in \cite{c18}. These capital allocation rules are not applicable to this study since they apply to linear portfolio operators, whereas we treat generic (not necessarily linear) portfolio operators. While \cite{ptm18} define Euler allocation rules for non-linear portfolios, they only apply to positive homogeneous risk measures with homogeneity degree equal to 1. The homogeneity degree of lambda quantiles, as we show in this paper, is however a function of the portfolio composition and the lambda function.

In this paper, we define risk contributions of lambda quantiles defined on the space of portfolio compositions as the partial derivatives of lambda quantiles with respect to asset units. We derive risk contributions of individual assets to the overall portfolio risk, measured via the lambda quantile of the portfolio composition. In doing so, we prove that lambda quantiles are continuously partially differentiable in the space of portfolio compositions using two independent methods which assume different properties. Furthermore, we prove that lambda quantiles are continuously differentiable in smaller subsets of $\mathbb{R}^n$, for a lambda function that may contain discontinuities, as long as it is continuously differentiable within a specific interval of $\mathbb{R}$.

Risk contributions calculated as directional derivatives of positive homogeneous risk measures of degree 1 of portfolios with one unit per asset are known as Euler contributions, where the assignment of capital using Euler contributions is known as Euler allocation. We show in this paper that lambda quantiles, scaled by a factor, can be written as a sum of their partial derivatives scaled by number of assets. This property is then used to show that lambda quantiles are homogeneous in the space of portfolio compositions, with a homogeneity degree that depends on both the portfolio composition and the lambda function. Only for the special case of a constant lambda function, the lambda quantile reduces to the VaR and has a homogeneity degree of 1. Therefore, the Euler allocation rule may not always be applicable to lambda quantiles, since their homogeneity degree is not universally equal to 1. Due to the variable nature of lambda quantiles' homogeneity degrees, we introduce a generalised Euler capital allocation rule, that is compatible with risk measures of any homogeneity degree and non-linear but homogeneous portfolios. We prove that the generalised Euler allocations of lambda quantiles have the full allocation property. We further provide a financial application using real market data that illustrates, for the case of non-linear portfolios, the lambda quantile homogeneity degree as a function of both the lambda function and the portfolio composition and the generalised Euler contributions of the portfolio assets. Notice that this notion of variable homogeneity degree is in favour of some criticisms that positive homogeneous risk measures of degree 1 defined on random variables have received. \cite{fs02}, for example, argue that large position multiples may induce additional liquidity risk, causing the portfolio risk to increase non-linearly compared to position size.

This paper builds upon methods and results relating to risk contributions and differentiability of VaR, whose literature is extensive and well established; indicatively see \cite{t99}, \cite{h03}, \cite{h09}, \cite{tm16}, \cite{s20}, and \cite{Pesenti2021RA}. Specifically, the papers of \cite{t99}, \cite{h09}, and \cite{tm16} provide a stepping stone for proving differentiability and calculating risk contributions of lambda quantiles from a portfolio performance measurement perspective, which are all novel pursuits in risk measure theory. Indeed, for lambda quantiles to be partially differentiable, we require additional smoothness assumptions. A first set of assumptions relates to the invertibility property of the portfolio profit and loss and the existence of an asset with a continuous density, similar to \cite{t99}. As this may not always be satisfied for generic portfolio profit and loss, we further prove our results using the condition that the portfolio profit and loss possesses a (locally) Lipschitz continuous, akin to the assumptions in \cite{h09}. 

The paper is organised as follows: Section \ref{sec:prelims} introduces the necessary notation and definitions. In Section \ref{sec:pmbrc}, we prove continuous partial differentiability of lambda quantiles in subsets of $\mathbb{R}^n$ and derive explicit formulae of risk contributions of lambda quantiles with respect to portfolio compositions. In Section \ref{sec:ed} we introduce the generalised Euler contributions and generalised Euler allocation rule and prove that the generalised Euler contributions of lambda quantiles fulfil the full allocation property. Section \ref{hpo} is devoted to the study of the homogeneity properties of generic portfolio operators. Section 6 illustrates the concept of homogeneity degree and risk contributions of lambda quantiles on a non-linear portfolio using financial market data.
\vspace{-1em}
\section{Preliminaries}\label{sec:prelims}

Let $(\Omega,\mathcal{F},\mathbb{P})$ be a probability space. We denote by $\mathcal{X}$ the set of random variables and by $\mathcal{X}^n$ for $n\geq2$ the set of random vectors on that space, taking values in $\mathbb{R}$ and $\mathbb{R}^n$ respectively. The joint probability distribution function of $\bm{X}=(X_1,\dots,X_n)\in\mathcal{X}^n$ is represented by $F_{\bm{X}}(\bm{x}):=\mathbb{P}(\bm{X}\leq\bm{x})$ for all $\bm{x}\in\mathbb{R}^n$, where each $X_1,\dots,X_n\in\mathcal{X}$. We will use $\bm{X}_{-1}:=(X_2,\dots,X_n)\in\mathcal{X}^{n-1}$ and $\bm{x}_{-1}:=(x_2,\dots,x_n)\in\mathbb{R}^{n-1}$ to indicate, respectively, random and ordinary vectors with first components removed. Define $\phi$ to be the density of the conditional probability distribution of $X_1$ given $X_2=x_2\,\dots,X_n=x_n$. Also, $U\subset\mathbb{R}\setminus\{0\}\times\mathbb{R}^{n-1}$ is a bounded set of $n$-dimensional real vectors with at least one non-zero component, which we set w.l.o.g. to the first component. Note that the choice of the first component is arbitrary and $\phi$ could represent the density of the conditional distribution of $X_i$ given $X_1=x_1,\dots,X_{i-1}=x_{i-1},X_{i+1}=x_{i+1},\dots,X_n=x_n$ for any $i=1,\dots,n$, provided that the $i^\text{th}$ component of $U$ does not contain zero.

In this paper, we treat a portfolio of $n$ assets. Random profits and losses of assets are represented by $\bm{X}$ and the portfolio composition is given by $\bm{u}\in U$.

\begin{definition}
A mapping $g:U\times\mathcal{X}^n\rightarrow\mathcal{X}$ is called a \emph{portfolio operator}.

For fixed $\bm{X}\in\mathcal{X}^n$, we call the mapping $g_{\bm{X}}:U\rightarrow\mathcal{X}$ such that $g_{\bm{X}}(\bm{u})=g[\bm{u},\bm{X}]$ the \emph{portfolio as a function of the composition $\bm{u}$} or \emph{portfolio} for short.

Finally, if the random vector $\bm{X}$ is realised, i.e. $\bm{X}(\omega)=\bm{x}\in\mathbb{R}^n$ for some outcome $\omega\in\Omega$, then we denote the portfolio using the mapping $g_{\bm{x}}:U\rightarrow\mathbb{R}$ such that $g_{\bm{x}}(\bm{u})=g[\bm{u},\bm{X}(\omega)]$, and call it the \emph{realised portfolio}.
\end{definition}
A portfolio operator  $g:U\times\mathcal{X}^n\rightarrow\mathcal{X}$ may represent the mapping from a composition $\bm{u}$ and a profit and loss vector $\bm{X}$ to the overall random portfolio's profit and loss. For fixed $\bm{X}\in\mathcal{X}^n$, $g_{\bm{X}}(\bm{u})$ can then be viewed as the \emph{portfolio profit and loss}. Note, that we do not require a portfolio to be linear in $\bm{X}$, indeed, the main focus of this paper is on non-linear portfolios.  Finally, if the random vector $\bm{u}$ is realised, i.e. $\bm{X}(\omega)=\bm{x}\in\mathbb{R}^n$ for some outcome $\omega\in\Omega$, then  $g_{\bm{x}}(\bm{u})=g[\bm{u},\bm{X}(\omega)]$, is the \emph{realised portfolio profit and loss}.

The portfolio operator $g$ is subject to stochastic variability because the value taken by $\bm{X}$ at each outcome $\omega\in\Omega$ is random and we assume that $g$ is independent of the probability distribution $F_{\bm{X}}$. The portfolio operator is also subject to distributional variability because we consider all random vectors in $\mathcal{X}^n$ -- we are not restricted to a class of random vectors of a specific distribution. If the random vector $\bm{X}$, and hence its joint probability distribution $F_{\bm{X}}$, is fixed, then the portfolio operator is only subject to stochastic variability and we consider the portfolio $g_{\bm{X}}$. Note that $g_{\bm{X}}(\bm{u})$, for any $\bm{u} \in U$, is a random variable, because $\bm{X}$ has not been realised. Moreover, $g_{\bm{X}}(\bm{u})$ varies (deterministically) with the dynamics of portfolio composition $\bm{u}$. As we must distinguish between the joint probability distribution $F_{\bm{X}}$ and the probability distribution function of the portfolio  $Y:=g_{\bm{X}}(\bm{u})$, we denote the probability distribution and density functions of the portfolio $Y$ by $F_{Y}(y)=\mathbb{P}(g_{\bm{X}}(\bm{u})\leq y)$ and $f_{Y}(y)=d F_{Y}/dy$, respectively, for all $y\in\mathbb{R}$. 

The following example shows the conceptual difference between the portfolio operator $g$ and the portfolio  $g_{\bm{X}}$.

\begin{example}
Let $\bm{X}=(X_1,X_2)$ and $\bm{u}=(u_1,u_2)$. Consider the portfolio operator:
\begin{equation*}
g[\bm{u},\bm{X}]=u_1X_1+u_2X_2-\mathbb{E}[u_1X_1+u_2X_2].
\end{equation*}
This operator represents the difference between actual and expected profits and losses of a portfolio, or in other words, the unexpected profit and loss. Even though the portfolio  $g_{\bm{X}}$ has the same form as the operator $g$, they are fundamentally different objects and we may choose to write $g_{\bm{X}}$ as:
\begin{equation*}
g_{\bm{X}}(\bm{u})=u_1X_1+u_2X_2-\mu_Y,
\end{equation*}
where $\mu_Y$ is the mean of the random variable $Y:=u_1X_1+u_2X_2$ with fixed $\bm{X}$. This is because for a fixed composition $\hat{\bm{u}}:=(\hat{u}_1,\hat{u}_2)$ and fixed $\bm{X}$, the distribution of $\hat{Y}:=\hat{u}_1X_1+\hat{u}_2X_2$ is also fixed and therefore the mean $\mu_{\hat{Y}}$ is a constant. On the other hand, if we do not fix $\bm{X}$, then the expectation $\mathbb{E}[\hat{u}_1X_1+\hat{u}_2X_2]$, that appears in the operator $g[\hat{\bm{u}},\bm{X}]$, is a function of $\bm{X}$.
\end{example}

In practice, if the distribution of asset is known, practitioners are interested in changing the portfolio composition $\bm{u}$ to achieve a higher risk-adjusted profit and loss for their portfolio. The process of selecting assets by comparing their expected profits and losses and contribution to overall portfolio risk is known as risk-adjusted performance measurement. In order to do this, one must know the per-unit contribution of each asset to the overall portfolio risk, and, in particular, risk contributions that are suitable for performance measurement.

In order to measure the risk, we use lambda quantiles that are traditionally defined on distributions. However, the purpose of this paper is to calculate the per-unit risk contribution of each asset to the overall portfolio risk. Hence, as one of the novelties of this paper, we define lambda quantiles on the set of the portfolio compositions $U$ and calculate partial derivatives of lambda quantiles with respect to asset units. The partial derivative with respect to asset units is the only definition of risk contribution that is suitable for performance measurement \citep{t99}.

\begin{definition}\label{def:lq}
The \emph{lambda quantile} $\rho_{\Lambda}:U\rightarrow\mathbb{R}$ with respect to $g_{\bm{X}}(\bm{u})\in\mathcal{X}$ is defined as follows:
\begin{equation*}
\rho_{\Lambda}(\bm{u};g_{\bm{X}}):=-\inf\{y\in\mathbb{R}\,|\,\mathbb{P}(g_{\bm{X}}(\bm{u})\leq y)>\Lambda(y)\}\,,
\end{equation*}
where $\Lambda:\mathbb{R}\rightarrow[\lambda_m,\lambda_M]$ is bounded such that $0<\lambda_m\leq\lambda_M<1$ and referred to as the \emph{lambda function}.
\end{definition}
The requirement that $\lambda_m$ and $\lambda_M$ are bounded away from 0 and 1, respectively, guarantees that the lambda quantile $\rho_{\Lambda}$ takes only finite values \citep{fmp14,bpr17}. 
In \cite{fmp14} the authors allow for $\lambda_m=0$ and $\lambda_M=1$ in their definition of lambda quantiles on the set of probability measures. However, the authors assume $\lambda_M<1$ within a financial context and provide a discussion on the choices of $\lambda_m$ being strictly positive; we refer the reader to Chapter 4 and Remark 19 of \cite{fmp14} for a detailed discussion on choices of $\lambda_m$ and $\lambda_M$.

The lambda quantile at $\bm{u}$ is the negative of the smallest intersection point of the distribution $F_Y$ and the lambda function $\Lambda$, provided they are both continuous. Otherwise, it is the negative of the smallest point from which the distribution $F_Y$ dominates the lambda function. As we work with profits and losses, we use the right quantiles of a distribution function and the negative of the right quantile. Left quantiles are typically used when asset losses are considered, e.g., in an insurance context. The lambda quantile $\rho_{\Lambda}(\bm{u};g_{\bm{X}})$ represents a positive amount to be allocated to absorb losses. A positive amount is allocated only if the profit and loss is negative (or loss is positive), i.e. $-\rho_{\Lambda}(\bm{u};g_{\bm{X}})<0$, otherwise the risk measure suggests there is a surplus of money which can be removed from the portfolio and still ensures its solvency. For the time being, we assume that the lambda function is bounded and we denote the derivative of the lambda function by $\Lambda'(x):=d\Lambda/dx$. We will, however, assume additional properties in subsequent sections.

From Definition \ref{def:lq}, we observe that the lambda quantile is a generalisation of the quantile function, in that the lambda quantile is the negative of the quantile function at a level determined by the lambda function $\Lambda$.
If the lambda function is a constant, i.e. $\Lambda(y)=\lambda\in(0,1)$ for all $y\in\mathbb{R}$, then the lambda quantile simplifies to the well-known  \emph{Value-at-Risk} (VaR). In particular, the $VaR_{\lambda}:U\rightarrow\mathbb{R}$ is given by:
\begin{equation*}
VaR_{\lambda}(\bm{u};g_{\bm{X}}):=-\inf\{y\in\mathbb{R}|\mathbb{P}(g_{\bm{X}}(\bm{u})\leq y)>\lambda\}.
\end{equation*}
Here, we view the VaR at (fixed) level $\lambda\in(0,1)$ as a function of $\bm{u}$, which is in contrast to the typically definition of VaR as a function of the random variable $g_{\bm{X}}(\bm{u})$. For simplicity, we write $\rho_{\Lambda}(\bm{u})$ and $VaR_{\lambda}(\bm{u})$ when there is no ambiguity on the portfolio  $g_{\bm{X}}$. 

We now provide examples of portfolio operators which we will use as running examples. Operators (\ref{mj:i})-(\ref{mj:iii}) are studied in \cite{m18}, which we have adapted to our notation.

\begin{example}\label{exmp:ops}
Let $n=2$, $\bm{X}=(X_1,X_2)$, $\bm{u}=(u_1,u_2)$ and $Y:=u_1X_1+u_2X_2$. The following are examples of portfolio operators defined on the Cartesian product $U\times\mathcal{X}^2$:
\begin{enumerate}[(i)]
\item $g[\bm{u},\bm{X}]=u_1X_1+u_2X_2$\label{op:lin}
\item $g[\bm{u},\bm{X}]=u_1X_1+u_2X_2-\mathbb{E}[u_1X_1+u_2X_2]$\label{mj:i}
\item $g[\bm{u},\bm{X}]=\max\{0,u_1X_1+u_2X_2-\mathbb{E}[u_1X_1+u_2X_2]\}$\label{mj:ii}
\item $g[\bm{u},\bm{X}]=u_1X_1+u_2X_2-VaR_{\lambda}(\bm{u};Y)$\label{mj:iii}

\item $g[\bm{u},\bm{X}]
= u_1X_1+u_2X_2- 
\min\{\max\{u_1X_1+u_2X_2 - D, 0\}, L\}$, with $D,L >0$\label{mj:iv}

\item $g[\bm{u},\bm{X}]=u_1^{\tau}X_1+u_2^{\tau}X_2-\rho_{\Lambda}(\bm{u};u_1^{\tau}X_1+u_2^{\tau}X_2)$, where ${\tau}\in\mathbb{R}$. \label{op:lq}
\end{enumerate}
\end{example}

Portfolio operators \eqref{mj:ii} and \eqref{mj:iv} are typical insurance portfolios. In particular, Example \ref{exmp:ops} \eqref{mj:iv} is the loss for an insurance company after reinsurance on $u_1 X_1 + u_2X_2$ with deductible $D$ and limit $L$.

Both the nature of assets in a portfolio and the context in which these assets are used determine the operator $g$. For example, $g$ may be determined by the pricing functions of assets. A simple linear portfolio of stock positions will only involve profits and losses of these linear assets. However, a more complex portfolio with both stock positions and stock options will include option profit and loss that are calculated using an option pricing formula (for example, the Black–Scholes formula), which is a non-linear function of the underlying stock's profit and loss. stock options or credit default swaps are examples of non-linear assets.

A portfolio operator is referred to as a \emph{linear operator} if it is a linear function of asset profits and losses, i.e. for all $\bm{X},\hat{\bm{X}}\in\mathcal{X}^n$ and all $c_1,c_2\in\mathbb{R}$ we have $g[\bm{u},c_1\bm{X}+c_2\hat{\bm{X}}]=c_1g[\bm{u},\bm{X}]+c_2g[\bm{u},\hat{\bm{X}}]$. Clearly, only portfolio operators (\ref{op:lin}) and (\ref{mj:i}) in Example \ref{exmp:ops} are linear. If there is at least one position which is a non-linear function of an underlying asset price, i.e. there is at least one non-linear asset, then the portfolio is \emph{non-linear} when viewed as a function of the underlying asset price.
Therefore, whether or not the operators given in Example \ref{exmp:ops} represent linear portfolios depend on what positions the $X_i$'s represent. For example, if both $X_1$ and $X_2$ in Example \ref{exmp:ops} (\ref{op:lin}) represent profits and losses of stocks, then this is an example of both a linear portfolio and linear operator. If, on the other hand, $X_1$ is a stock's profit and loss and $X_2$ is the profit and loss of an option, then whilst the operator and the portfolio viewed as a function of $\bm{X}$ are linear, the operator and the portfolio are non-linear when viewed as a function of $X_1$ and the profit and loss of the option's underlying asset. If not otherwise stated, we assume that a portfolio is linear if its corresponding portfolio operator is linear.

For the exposition, we require additional notation. Define $(x_1,\bm{x}_{-1}):=\bm{x}$, so that the realised portfolio, for all $\omega\in\Omega$ with $\bm{X}(\omega)=\bm{x}$, becomes:
\begin{equation*}
g_{\bm{X}(\omega)}(\bm{u})=g_{\bm{x}}(\bm{u})=g_{(x_1,\bm{x}_{-1})}(\bm{u})\,.
\end{equation*}
Further, the partial derivatives of $g_{\bm{x}}$ with respect to $u_i$ and $x_1$ are denoted by $\partial_{u_i} g_{\bm{x}}(\bm{u}):=\partial g_{\bm{x}}/\partial u_i$ and $\partial_{x_1} g_{\bm{x}}(\bm{u}):=\partial g_{\bm{x}}/\partial x_1$, respectively. Similarly, we denote the $\mathbb{P}$-a.s. partial derivatives of $g_{\bm{X}}$ with respect to $u_i$ by $\partial_{u_i} g_{\bm{X}}(\bm{u}):=\partial g_{\bm{X}}/\partial u_i$.

We write $A_1\subset\mathbb{R}$ for the support of the random variable $X_1$, i.e. $A_1\subset\mathbb{R}$ is the smallest closed set such that $\mathbb{P}(X_1\in A_1)=\mathbb{P}(\{\omega\in\Omega~|~X_1(\omega)\in A_1\})=1$. We say that $g_{\bm{X}}$ is \emph{invertible with respect to $X_1=x_1$, for all $X_2=x_2,\dots,X_n=x_n$}, if for all $x_1\in A_1$, the function $g_{(x_1,\bm{x}_{-1})}$ is invertible with respect to $x_1$, for all $x_2,\dots,x_n$. We denote the inverse of $g_{(x_1,\bm{x}_{-1})}$ by $l_{(y,\bm{x}_{-1})}:U\rightarrow\mathbb{R}$ such that
\begin{equation*}
l_{(y,\bm{x}_{-1})}(\bm{u})=x_1 \iff g_{(x_1,\bm{x}_{-1})}(\bm{u})=y,
\end{equation*}
for all $x_1\in A_1$, and its partial derivatives with respect to $u_i$ and $y$ are respectively $\partial_{u_i} l_{(y,\bm{x}_{-1})}(\bm{u}):=\partial l_{(y,\bm{x}_{-1})}/\partial u_i$ and $\partial_{y} l_{(y,\bm{x}_{-1})}(\bm{u}):=\partial l_{(y,\bm{x}_{-1})}/\partial y$. By Assumption \ref{asmp:t} \eqref{asmp:mi} below, the realised portfolio $g_{\bm{x}}$ is strictly increasing in $x_1$ and thus in this case the inverse $l$ is well-defined.

\section{Differentiability and risk contributions of lambda quantiles}\label{sec:pmbrc}

In this section, we study differentiability of lambda quantiles in the set $U$. Derivatives of the lambda quantile risk measure have not been studied in previous literature. Derivatives of the special case, namely the VaR measure, however, have an extensive literature, see for example, \cite{t99, gls00, h03, h09, tm16}. Although these studies calculate derivatives of VaR, they differ in both methods used and assumptions made in their respective settings.

Partial derivatives of risk measures with respect to asset units are crucial in portfolio risk management as they represent the risk contribution of each asset to the overall portfolio risk. This definition of risk contributions is consistent with risk-adjusted performance measurement of portfolios (see Definition \ref{defn:pm} and Lemma \ref{lem:tspm} in the Appendix \ref{app:results}). This section extends the literature to include partial derivatives of lambda quantiles and recover previous results on VaR as special cases. In light of this motivation, this section has two objectives. The first objective is to provide conditions under which lambda quantiles are continuously partially differentiable in the set $U$. The second objective is to calculate these partial derivatives explicitly. There are several approaches one may follow to achieve the latter, which ultimately depend on assumptions made regarding the portfolio $g_{\bm{X}}$, the random vector $\bm{X}$, and the lambda function $\Lambda$. Partial derivatives of lambda quantiles will be calculated using two different approaches, each having its own set of assumptions.

In the first approach we generalise the treatment in \cite{t99}, who calculated partial derivatives of $VaR_{\lambda}$ for linear portfolios, to lambda quantiles for generic portfolios. We extend this method to take into account the lambda function (instead of a fixed level $\lambda$) and to cover non-linear portfolios (by defining lambda quantiles on generic portfolios). These generalisations require additional assumptions for lambda quantiles to be continuously differentiable in the set $U$.

In the second approach we utilise the closed-form representation of probability sensitivities, proposed by \cite{h09}. The probability sensitivity corresponds, in our context, to the partial derivative of the portfolio's probability distribution function with respect to asset units. Note that the two approaches mentioned above allow us to prove the same property (continuously partially differentiable in $U$) of lambda quantiles. Furthermore, partial derivatives of lambda quantiles are the same under both approaches.

\begin{asmp}\label{asmp:t}
We say that Assumption \ref{asmp:t} is satisfied if:
\begin{enumerate}[(i)]
\item $g_{(x_1,\bm{x}_{-1})}(\bm{u})$ is strictly increasing in $x_1$, for all $\bm{u} \in U$. \label{asmp:mi}

\item $g_{\bm{X}}(\bm{u})$ is $\mathbb{P}$-a.s. differentiable in $\bm{u}$, for all $\bm{u}\in U$. \label{asmp:diff}

\item For fixed $\bm{x}_{-1}$, the density $y\mapsto\phi(y|\bm{x}_{-1})$ is continuous in $y$.\label{asmp:dens}

\item The function $l_{(y,\bm{x}_{-1})}(\bm{u})$ is continuously differentiable with respect to $y$ and $\bm{u}$. \label{asmp:cinv}

\item For fixed $\bm{u}$ and all $i=1,\dots,n$, the following maps are uniformly bounded with respect to $y$: \label{asmp:fint}
\begin{align*}
& y \mapsto \mathbb{E}[\partial_y l_{(y,\bm{X}_{-1})}(\bm{u})\phi(l_{(y,\bm{X}_{-1})}(\bm{u})|\bm{X}_{-1})],\\
&y\mapsto\mathbb{E}[\partial_{u_i} l_{(y,\bm{X}_{-1})}(\bm{u})\phi(l_{(y,\bm{X}_{-1})}(\bm{u})|\bm{X}_{-1})].
\end{align*}

\item For $g_{(x_1,\bm{x}_{-1})}(\bm{u})$ strictly increasing in $x_1$, for each $\bm{u}\in U$, assume:\label{asmp:elq}
\begin{equation*}
\mathbb{E}[\partial_y l_{(-\rho_{\Lambda}(\bm{u}),\bm{X}_{-1})}(\bm{u})\phi(l_{(-\rho_{\Lambda}(\bm{u}),\bm{X}_{-1})}(\bm{u})|\bm{X}_{-1})]>\Lambda'(-\rho_{\Lambda}(\bm{u})).
\end{equation*}
\end{enumerate}
\end{asmp}

Note that Assumption \ref{asmp:t} (\ref{asmp:mi}) implies that $g_{\bm{X}}(\bm{u})$ is $\mathbb{P}$-a.s. invertible with respect to $X_1=x_1$, for all $X_2=x_2,\dots,X_n=x_n$. Although we assume that $g_{(x_1,\bm{x}_{-1})}(\bm{u})$ is strictly increasing in $x_1$ in this paper, the results also hold for the strictly decreasing case, with some sign changes. The proofs of the decreasing case are similar to those of the increasing case and thus omitted. If Assumption \ref{asmp:t} \eqref{asmp:mi} to \eqref{asmp:elq} are fulfilled, then by Lemma \ref{lem:F} below, Assumption \ref{asmp:t} \eqref{asmp:elq} is equivalent to 
\begin{equation}\label{eq:asm}
    f_Y(- \rho_\Lambda(\bm{u})) > \Lambda^\prime(-\rho_\Lambda(\bm{u}))\,.
\end{equation}
We require Equation \eqref{eq:asm} to ensure that the portfolio density adjustment (see  Definition \ref{defn:rda}) is well-defined and positive at the point $y=-\rho_\Lambda(\bm{u})$. The portfolio density adjustment is fundamental for the homogeneity degree and the  risk contributions of lambda quantiles.

\begin{example}
Here we discuss which of the portfolios in Example \ref{exmp:ops} fulfil Assumption \ref{asmp:t}. Note that the corresponding portfolios of the  portfolio operator \eqref{op:lin}, \eqref{mj:i}, \eqref{mj:iii}, and \eqref{op:lq} are of the form $g_{\bm{X}}(\bm{u}) = u_1^\tau X_1 + u_2^\tau X_2 + c$, for a constant $c \in \mathbb R$ and $\tau \ge 0$. Thus, these portfolio operators satisfy, under suitable condition on the joint distribution of $\bm{X}$, Assumptions \ref{asmp:t}. The portfolio operators \eqref{mj:ii} and \eqref{mj:iv} do not fulfil Assumption \ref{asmp:t}, however, they satify Assumption \ref{asmp:h} below.
\end{example}

\begin{remark}\label{rem:asmpt}
Assumption \ref{asmp:t} (\ref{asmp:dens}) and (\ref{asmp:cinv}) imply the following mappings are continuous for fixed $\bm{x}_{-1}$ and $i=1,\dots,n$
\begin{align*}
&(y,\bm{u})\mapsto\partial_y l_{(y,\bm{x}_{-1})}(\bm{u})\phi(l_{(y,\bm{x}_{-1})}(\bm{u})|\bm{x}_{-1}),\\
&(y,\bm{u})\mapsto\partial_{u_i} l_{(y,\bm{x}_{-1})}(\bm{u})\phi(l_{(y,\bm{x}_{-1})}(\bm{u})|\bm{x}_{-1}),
\end{align*}
and the following mappings are continuous, for all $i=1,\dots,n$,
\begin{align*}
&(y,\bm{u})\mapsto \mathbb{E}[\partial_y l_{(y,\bm{X}_{-1})}(\bm{u})\phi(l_{(y,\bm{X}_{-1})}(\bm{u})|\bm{X}_{-1})],\\
&(y,\bm{u})\mapsto\mathbb{E}[\partial_{u_i} l_{(y,\bm{X}_{-1})}(\bm{u})\phi(l_{(y,\bm{X}_{-1})}(\bm{u})|\bm{X}_{-1})].
\end{align*}
\end{remark}

The second set of assumptions relates to the probability sensitivity of \cite{h09}, which have been adapted to our setting. In contrast to the approach taken in \cite{h09}, we require Assumption \ref{asmp:h} (\ref{asmp:grad}) to account for the lambda function in our treatment.
\begin{asmp}\label{asmp:h}
We say that Assumption \ref{asmp:h} is satisfied if:
\begin{enumerate}[(i)]

\item $g_{\bm{X}}(\bm{u})$ is $\mathbb{P}$-a.s. differentiable in $\bm{u}$, for all $\bm{u}\in U$. \label{asmp:h1a}

\item There exists a random variable $m(\bm{X})$ with $\mathbb{E}[m(\bm{X})]<\infty$ such that for all $\bm{u},\bm{v}\in U$: \label{asmp:h1b}
\begin{equation*}
|g_{\bm{X}}(\bm{u})-g_{\bm{X}}(\bm{v})|~\leq~ m(\bm{X})\|\bm{u}-\bm{v}\| \quad\mathbb{P}\text{-a.s.}\,,
\end{equation*}
where $\|\cdot\|$ denotes the Euclidean norm in $U$.

\item For all $\bm{u}\in U$, the random variable $g_{\bm{X}}(\bm{u})$ has a continuous density denoted by $f_Y(y)$ in a neighbourhood of $y=-\rho_{\Lambda}(\bm{u})$. \label{asmp:h2a}

\item For the function $F:\mathbb{R}\times U\rightarrow[0,1]$ defined as $F(y,\bm{u}):=\mathbb{P}(g_{\bm{X}}(\bm{u})\leq y)$, the partial derivatives $\partial_{u_i}F(y,\bm{u})$ exist and are continuous in $\bm{u}$ and $y$ in a neighbourhood of $y=-\rho_{\Lambda}(\bm{u})$, for all $i=1,\dots,n$. \label{asmp:h2b}
\item For all $\bm{u}\in U$ and for $i=1,\dots,n$, the following mappings are continuous at $y=-\rho_{\Lambda}(\bm{u})$: \label{asmp:h3}
\begin{equation*}
y\mapsto\mathbb{E}[\partial_{u_i} g_{\bm{X}}(\bm{u})~|~g_{\bm{X}}(\bm{u})=y].
\end{equation*}
\item $f_Y(-\rho_{\Lambda}(\bm{u}))>\Lambda'(-\rho_{\Lambda}(\bm{u}))$ for all $\bm{u}\in U$. \label{asmp:grad}
\end{enumerate}
\end{asmp}

Assumption \ref{asmp:h} (\ref{asmp:h2a}) implies that each distribution of the random field $(g_{\bm{X}}(\bm{u}))_{\bm{u}\in U}$ is continuous in a neighbourhood of $y=-\rho_{\Lambda}(\bm{u})$ for its respective $\bm{u}\in U$.

\begin{example}
The portfolios corresponding to the portfolio operators in Example \ref{exmp:ops} \eqref{op:lin}, \eqref{mj:i}, \eqref{mj:ii}, and \eqref{asmp:fint} clearly fulfil Assumptions \ref{asmp:h} \eqref{asmp:h1a} and \eqref{asmp:h1b}. The remaining portfolios in Example \ref{exmp:ops} do not in general fulfil Assumption \ref{asmp:h} \eqref{asmp:h1b}, as lambda quantiles and $VaR$ are not Lipschitz continuous.
\end{example}

Using Assumption \ref{asmp:t}, we demonstrate Lemmas \ref{lem:F} and \ref{lem:cexp}, which we require to prove Theorem \ref{thm:tlqrc} with condition \ref{thmasmp:t} and Proposition \ref{propo:tlqrci}. Lemmas \ref{lem:F} and \ref{lem:cexp} are generalisations of Lemmas 3.2 and 2.2 in \cite{t01}, respectively, as they apply to both linear and non-linear portfolios.

\begin{lemma}\label{lem:F}
Suppose Assumption \ref{asmp:t} (\ref{asmp:mi})-(\ref{asmp:fint}) are satisfied. Then the function $F:\mathbb{R}\times U\rightarrow[0,1]$ defined as $F(y,\bm{u}):=\mathbb{P}(g_{\bm{X}}(\bm{u})\leq y)$ is partially differentiable in $y$ and $u_i$, for $i=1,\dots,n$. The continuous derivatives are given by:
\begin{align}
\frac{\partial F}{\partial y}(y,\bm{u})&=\mathbb{E}\left[\partial_y l_{(y,\bm{X}_{-1})}(\bm{u})\phi(l_{(y,\bm{X}_{-1})}(\bm{u})|\bm{X}_{-1})\right],\label{eqn:pFy}
\\
\frac{\partial F}{\partial u_i}(y,\bm{u})&=\mathbb{E}\left[\partial_{u_i} l_{(y,\bm{X}_{-1})}(\bm{u})\phi(l_{(y,\bm{X}_{-1})}(\bm{u})|\bm{X}_{-1})\right]\label{eqn:pFu}.
\end{align}
\end{lemma}
\begin{proof}
We generalise the approach taken in the proof of Lemma 5.3 in \cite{t99} to prove that $F$ is continuously differentiable. Our method applies to a generic random variable $g_{\bm{X}}(\bm{u})$ whilst the proof provided in \cite{t99} applies only to linear portfolios, that is to $g_{\bm{X}}(\bm{u}) = \sum_{i=1}^n u_iX_i$.

We first introduce the following integral using the density $\phi$ of the conditional distribution of $X_1$ given $\bm{X}_{-1}=\bm{x}_{-1}$
\begin{equation}\label{eqn:G}
G(y,\bm{u},\bm{x}_{-1}):=\int_{-\infty}^{l_{(y,\bm{x}_{-1})}(\bm{u})}\phi(t|\bm{x}_{-1})dt.
\end{equation}
Note that $G$ can be written in the following form
\begin{align*}
G(y,\bm{u},\bm{x}_{-1})&=\mathbb{P}(X_1\leq l_{(y,\bm{x}_{-1})}(\bm{u})|\bm{X}_{-1}=\bm{x}_{-1})\\
&=\mathbb{P}(\{\omega\in\Omega|X_1(\omega)\leq l_{(y,\bm{x}_{-1})}(\bm{u})\})\\
&=\mathbb{P}(\{\omega\in\Omega|g_{(X_1(\omega),\bm{x}_{-1})}(\bm{u})\leq g_{(l_{(y,\bm{x}_{-1})},\bm{x}_{-1})}(\bm{u})\})\\
&=\mathbb{P}(\{\omega\in\Omega|g_{(X_1(\omega),\bm{x}_{-1})}(\bm{u})\leq y\})\\
&=\mathbb{P}(g_{\bm{X}}(\bm{u})\leq y|\bm{X}_{-1}=\bm{x}_{-1}),
\end{align*}
and $F$ can be written in terms of $G$:
\begin{equation}\label{lem:FG}
F(y,\bm{u})=\mathbb{E}[\mathbb{P}(g_{\bm{X}}(\bm{u})\leq y|\bm{X}_{-1})]=\mathbb{E}[G(y,\bm{u},\bm{X}_{-1})].
\end{equation}
We show that $F$ is continuously differentiable in $y$ and $u_i$, for $i=1,\dots,n$, and that its derivatives can be computed by changing the order of integration and differentiation on the right-hand side of (\ref{lem:FG}). In order to do this, we apply Lemma \ref{lem:d19} (see Appendix \ref{app:results}) to the function $G:\mathbb{R}\times U\times\mathbb{R}^{n-1}\rightarrow\mathbb{R}$ to the components $y$ and $u_1, \ldots, u_n$. For this, we define $S_y:=U\times\mathbb{R}^{n-1}$, $S_{u_1}:=U\setminus\{\mathbb{R}\setminus\{0\}\}\times\mathbb{R}\times\mathbb{R}^{n-1}$, and $S_{u_j}:=U\setminus\mathbb{R}\times\mathbb{R}\times\mathbb{R}^{n-1}$ for $j = 2,\dots,n$. Note that we distinguish $u_1$ from $u_2, \ldots, u_n$ as $u_1$ cannot be zero.

For differentiability in the first component $y$, condition (i) of Lemma \ref{lem:d19} is satisfied as:
\begin{equation*}
\int_{S_y}|G(y,\bm{u},\bm{x}_{-1})|\,dF_{\bm{X}_{-1}}(\bm{x}_{-1})\,d\bm{u}=\int_{U}\mathbb{E}[|G(y,\bm{u},\bm{X}_{-1})|]\,d\bm{u}=\int_UF(y,\bm{u})\,d\bm{u}<\infty.
\end{equation*}
The finiteness follows from the observation that for any fixed $\bm{u}\in U$, $F(y,\bm{u})=F_Y(y)$ is the distribution function of $Y=g_{\bm{X}}(\bm{u})$ and that $U$ is bounded.
The same condition is satisfied for differentiability of $u_1$ as:
\begin{align*}
\int_{S_{u_1}}|G(y,\bm{u},\bm{x}_{-1})|\,dF_{\bm{X}_{-1}, Y}(\bm{x}_{-1}, y)\,d\bm{u}_{-1}&=\int_{U\setminus\{\mathbb{R}\setminus\{0\}\}}\mathbb{E}[G(Y,\bm{u},\bm{X}_{-1})]\,d\bm{u}_{-1}\\
&=\int_{U\setminus\{\mathbb{R}\setminus\{0\}\}]}\mathbb{E}[F(Y,\bm{u})]\,d\bm{u}_{-1}<\infty,
\end{align*}
where $\bm{u}_{-1}:=(u_2,\dots,u_n)\in U\setminus\{\mathbb{R}\setminus\{0\}\}$ and $F_{\bm{X}_{-1}, Y}:\mathbb{R}^{n-1}\times\mathbb{R}\rightarrow\mathbb{R}$ is the distribution function of the joint probability distribution of $\bm{X}_{-1}$ and $Y$. Similarly, for $u_2,\dots,u_n$ we have:
\begin{align*}
\int_{S_{u_j}}|G(y,\bm{u},\bm{x}_{-1})|\,dF_{\bm{X}_{-1}, Y}(\bm{x}_{-1}, y)\,d\bm{u}_{-j}&=\int_{U\setminus\mathbb{R}}\mathbb{E}[G(Y,\bm{u},\bm{X}_{-1})]\,d\bm{u}_{-j}\\
&=\int_{U\setminus\mathbb{R}}\mathbb{E}[F(Y,\bm{u})]\,d\bm{u}_{-j}<\infty,
\end{align*}
where for $j>1$, $\bm{u}_{-j}:=(u_1,\dots,u_{j-1},u_{j+1},\dots,u_n)\in U\setminus\mathbb{R}$. For condition (ii) of Lemma \ref{lem:d19}, we differentiate $G$ partially with respect to $y$ and $u_i$, $i=1,\dots,n$ using \eqref{eqn:G}:
\begin{align*}
\frac{\partial G}{\partial y}(y,\bm{u},\bm{x}_{-1})&=\partial_y l_{(y,\bm{x}_{-1})}(\bm{u})\phi(l_{(y,\bm{x}_{-1})}(\bm{u})|\bm{x}_{-1})\,,\\
\frac{\partial G}{\partial u_i}(y,\bm{u},\bm{x}_{-1})&=\partial_{u_i}l_{(y,\bm{x}_{-1})}(\bm{u})\phi(l_{(y,\bm{x}_{-1})}(\bm{u})|\bm{x}_{-1}),
\end{align*}
which are all continuous for fixed $\bm{x}_{-1}$ by Remark \ref{rem:asmpt}. For condition (iii) of Lemma \ref{lem:d19}, observe that the integral of $\partial G/\partial y$ in the domain $S_y$ is continuous in $y$ by Remark \ref{rem:asmpt}:
\begin{align*}
\int_{S_y}\frac{\partial G}{\partial y}(y,\bm{u},\bm{x}_{-1})\,dF_{\bm{X}_{-1}}(\bm{x}_{-1})\,d\bm{u}&=\int_U\mathbb{E}\biggl[\frac{\partial G}{\partial y}(y,\bm{u},\bm{X}_{-1})\biggl]\,d\bm{u}\\
&=\int_U\mathbb{E}[\partial_y l_{(y,\bm{X}_{-1})}(\bm{u})\phi(l_{(y,\bm{X}_{-1})}(\bm{u})|\bm{X}_{-1})]\,d\bm{u}.
\end{align*}
Similarly, by Remark \ref{rem:asmpt} integral of $\partial G/\partial u_1$ in $S_{u_1}$ is continuous in $u_1$:
\begin{align*}
\int_{S_{u_1}}\frac{\partial G}{\partial u_1}(y,\bm{u},\bm{x}_{-1})\,dF_{\bm{X}_{-1}, Y}(\bm{x}_{-1}, y)\,d\bm{u}_{-1}
&=
\int_{U\setminus\{\mathbb{R}\setminus\{0\}\}}\mathbb{E}\biggl[\frac{\partial G}{\partial u_1}(Y,\bm{u},\bm{X}_{-1})\biggl]\,d\bm{u}_{-1}
\\
&=
\int_{U\setminus\{\mathbb{R}\setminus\{0\}\}}\mathbb{E}[\partial_{u_1}l_{(Y,\bm{X}_{-1})}(\bm{u})\phi(l_{(Y,\bm{X}_{-1})}(\bm{u})|\bm{X}_{-1})]\,d\bm{u}_{-1},
\end{align*}
and the integral of $\partial G/\partial u_j$ in $S_{u_j}$ is continuous in $u_j$ for $j=2,\dots,n$:
\begin{align*}
\int_{S_{u_j}}\frac{\partial G}{\partial u_j}(y,\bm{u},\bm{x}_{-1})\,dF_{\bm{X}_{-1}, Y}(\bm{x}_{-1}, y)\,d\bm{u}_{-j}&=\int_{U\setminus\mathbb{R}}\mathbb{E}\biggl[\frac{\partial G}{\partial u_j}(Y,\bm{u},\bm{X}_{-1})\biggl]\,d\bm{u}_{-j}\\
&=\int_{U\setminus\mathbb{R}}\mathbb{E}[\partial_{u_j}l_{(Y,\bm{X}_{-1})}(\bm{u})\phi(l_{(Y,\bm{X}_{-1})}(\bm{u})|\bm{X}_{-1})]\,d\bm{u}_{-j}.
\end{align*}
For condition (iv) of Lemma \ref{lem:d19}, we show that for $\delta>0$, the following integrals are finite:
\begin{align*}
I_y:&=\int_{S_y}\int_{-\delta}^{\delta}\biggl|\frac{\partial G}{\partial y}(y+\theta,\bm{u},\bm{x}_{-1})\biggl|\,d\theta\,dF_{\bm{X}_{-1}}(\bm{x}_{-1})\,d\bm{u},\\
I_{u_1}:&=\int_{S_{u_1}}\int_{-\delta}^{\delta}\biggl|\frac{\partial G}{\partial u_1}(y,\bm{u}+\theta\bm{e}_1,\bm{x}_{-1})\biggl|\,d\theta\,dF_{\bm{X}_{-1}, Y}(\bm{x}_{-1}, y)\,d\bm{u}_{-1},\\
I_{u_j}:&=\int_{S_{u_j}}\int_{-\delta}^{\delta}\biggl|\frac{\partial G}{\partial u_j}(y,\bm{u}+\theta\bm{e}_j,\bm{x}_{-1})\biggl|\,d\theta\,dF_{\bm{X}_{-1}, Y}(\bm{x}_{-1}, y)\,d\bm{u}_{-j},
\end{align*}
for $j=2,\dots,n$ where $(\bm{e}_k)_l=1$ for $k=l$ and $0$ otherwise. Recall that $\phi$ is a continuous density by Assumption \ref{asmp:t} (\ref{asmp:dens}) and $\partial_y l_{(y,\bm{x}_{-1})}(\bm{u})\geq 0$ because $l_{(y,\bm{x}_{-1})}$ is strictly increasing in $y$. Therefore,
\begin{equation*}
\frac{\partial G}{\partial y}(y+\theta,\bm{u},\bm{x}_{-1})=\partial_y l_{(y+\theta,\bm{x}_{-1})}(\bm{u})\phi(l_{(y+\theta,\bm{x}_{-1})}(\bm{u})|\bm{x}_{-1})\ge0,
\end{equation*}
for all $y\in\mathbb{R}$, $\theta\in(-\delta,\delta)$, $\bm{u}\in U$ and $\bm{x}_{-1}\in\mathbb{R}^{n-1}$. We can also see the above inequality by noting that $G$ is the conditional probability distribution of $Y$ given $\bm{X}_{-1}$, and therefore its partial derivative w.r.t $y$ is a conditional probability density. Now, observe that:
\begin{align*}
I_y&=\int_{S_y}\int_{-\delta}^{\delta}\frac{\partial G}{\partial y}(y+\theta,\bm{u},\bm{x}_{-1})\,d\theta\,dF_{\bm{X}_{-1}}(\bm{x}_{-1})\,d\bm{u}
\\
&=\int_{S_y}(G(y+\delta,\bm{u},\bm{x}_{-1})-G(y-\delta,\bm{u},\bm{x}_{-1}))\,dF_{\bm{X}_{-1}}(\bm{x}_{-1})\,d\bm{u}
\\
&=\int_{U}\mathbb{E}[G(y+\delta,\bm{u},\bm{X}_{-1})-G(y-\delta,\bm{u},\bm{X}_{-1})]\,d\bm{u}
\\
&=\int_{U}\biggl(F(y+\delta,\bm{u})-F(y-\delta,\bm{u})\biggl)\,d\bm{u}<\infty.
\end{align*}
To prove $I_{u_1}$ is finite, observe that since the integrand is positive, we can change the order of integration by Tonelli's theorem:
\begin{align*}
I_{u_1}&=\int_{-\delta}^{\delta}\int_{S_{u_1}}\biggl|\frac{\partial G}{\partial u_1}(y,\bm{u}+\theta\bm{e}_1,\bm{x}_{-1})\biggl|\,dF_{\bm{X}_{-1},Y}(\bm{x}_{-1},y)\,d\bm{u}_{-1}\,d\theta
\\
&=
\int_{-\delta}^{\delta}\int_{S_{u_1}}\max\biggl\{\frac{\partial G}{\partial u_1}(y,\bm{u}+\theta\bm{e}_1,\bm{x}_{-1}),-\frac{\partial G}{\partial u_1}(y,\bm{u}+\theta\bm{e}_1,\bm{x}_{-1})\biggl\}\,dF_{\bm{X}_{-1},Y}(\bm{x}_{-1},y)\,d\bm{u}_{-1}\,d\theta
\\
&=
\int_{-\delta}^{\delta}\int_{U\setminus\{\mathbb{R}\setminus\{0\}\}}\mathbb{E}\biggl[\mathbb{E}\biggl[\max\biggl\{\frac{\partial G}{\partial u_1}(y,\bm{u}+\theta\bm{e}_1,\bm{X}_{-1}),-\frac{\partial G}{\partial u_1}(y,\bm{u}+\theta\bm{e}_1,\bm{X}_{-1})\biggl\}~\bigg| ~Y = y\biggl]\biggl]\,d\bm{u}_{-1}\,d\theta<\infty,
\end{align*}
where we note that by Assumption \ref{asmp:t} (\ref{asmp:fint}), the conditional expectation is bounded. The proof of $I_{u_j}<\infty$ for $j=2,\dots,n$ follows the same approach. Therefore, by  Lemma \ref{lem:d19} we conclude that $F$ is continuously partially differentiable in $y$ and $u_i$ for $i=1,\dots,n$ with derivatives:
\begin{equation*}
\frac{\partial F}{\partial y}(y,\bm{u})=\mathbb{E}\biggl[\frac{\partial G}{\partial y}(y,\bm{u},\bm{X}_{-1})\biggl]\quad\text{ and }\quad\frac{\partial F}{\partial u_i}(y,\bm{u})=\mathbb{E}\biggl[\frac{\partial G}{\partial u_i}(y,\bm{u},\bm{X}_{-1})\biggl].
\end{equation*}
\end{proof}

\begin{remark}
Lemma \ref{lem:F} also holds if $g_{(x_1,\bm{x}_{-1})}$ is strictly decreasing in $x_1$ with a change of sign of the partial derivatives in Equations \eqref{eqn:pFy} and \eqref{eqn:pFu}. The proof of the decreasing case follows using similar steps as the increasing case, and noting that for the decreasing case: 
\begin{equation*}
G(y,\bm{u},\bm{x}_{-1})=\mathbb{P}(X_1\leq l_{(y,\bm{x}_{-1})}(\bm{u})|\bm{X}_{-1}=\bm{x}_{-1})=1-\mathbb{P}(g_{\bm{X}}(\bm{u})\leq y|\bm{X}_{-1}=\bm{x}_{-1}),
\end{equation*}
and again $F$ can be written in terms of $G$:
\begin{equation*}
F(y,\bm{u})=1-\mathbb{E}[G(y,\bm{u},\bm{X}_{-1})].
\end{equation*}
Finally, to prove that $I_y<\infty$, we note that:
\begin{equation*}
\biggl|\frac{\partial G}{\partial y}(y+\theta,\bm{u},\bm{x}_{-1})\biggl|=-\frac{\partial G}{\partial y}(y+\theta,\bm{u},\bm{x}_{-1}),
\end{equation*}
since $\partial_y l_{(y+\theta,\bm{x}_{-1})}(\bm{u})\leq0$ for all $y\in\mathbb{R}$, $\theta\in(-\delta,\delta)$, $\bm{u}\in U$, and $\bm{x}\in\mathbb{R}^{n-1}$.
\end{remark}

\begin{remark}\label{rem:dens}
If Assumption \ref{asmp:t} (\ref{asmp:mi})-(\ref{asmp:fint}) are satisfied then, for any $\bm{u}\in U$, the random variable $Y=g_{\bm{X}}(\bm{u})$ has a continuous probability density function given by:
\begin{equation*}
f_Y(y)=\mathbb{E}[\partial_y l_{(y,\bm{X}_{-1})}(\bm{u})\phi( l_{(y,\bm{X}_{-1})}(\bm{u})|\bm{X}_{-1})].
\end{equation*}
To see this, note that for fixed $\bm{u}\in U$, $F$ and $F_Y$ are identical, i.e. $F(y,\bm{u})=F_Y(y)$ for all $y\in\mathbb{R}$. The continuous partial derivative of $F$ with respect to $y$ is then given in Lemma \ref{lem:F} as:
\begin{equation*}
\frac{\partial F}{\partial y}(y,\bm{u})=\frac{\partial F_Y}{\partial y}(y)=f_Y(y)=\mathbb{E}[\partial_y l_{(y,\bm{X}_{-1})}(\bm{u})\phi(l_{(y,\bm{X}_{-1})}(\bm{u})|\bm{X}_{-1})].
\end{equation*}
Therefore, we see that Assumption \ref{asmp:t} (\ref{asmp:elq}) corresponds to the gradient of the distribution function $F_Y$ being greater than that of the lambda function at the point $y=-\rho_{\Lambda}(\bm{u})$.

The strictly decreasing case follows from the same argument with a change of sign. Note that the density $f_Y(y)$ is indeed positive for the decreasing case, because $l_{(y,\bm{x}_{-1})}$ is decreasing in $y$ and therefore $\partial_y l_{(y,\bm{x}_{-1})}(\bm{u})\leq0$ for all $y\in\mathbb{R}$, $\bm{u}\in U$, and $\bm{x}_{-1}\in\mathbb{R}^{n-1}$.
\end{remark}

\begin{lemma}\label{lem:cexp}
If Assumption \ref{asmp:t} (\ref{asmp:mi})-(\ref{asmp:fint}) are satisfied then, for any $\bm{u}\in U$ and $i=1,\dots,n$, we have:
\begin{equation}
\mathbb{E}[\partial_{u_i} g_{\bm{X}}(\bm{u})~|~g_{\bm{X}}(\bm{u})=y]=-\frac{\mathbb{E}[\partial_{u_i} l_{(y,\bm{X}_{-1})}(\bm{u})\phi(l_{(y,\bm{X}_{-1})}(\bm{u})|\bm{X}_{-1})]}{\mathbb{E}[\partial_y l_{(y,\bm{X}_{-1})}(\bm{u})\phi(l_{(y,\bm{X}_{-1})}(\bm{u})|\bm{X}_{-1})]}.
\end{equation}
\end{lemma}
\begin{proof}
The proof method is inspired by the proof of Lemma 1 in \cite{tm16}. Our proof, however, considers a portfolio $g_{\bm{X}}(\bm{u})$ on the set $U$, whereas \cite{tm16} do not use asset units. 

Consider the following expectation for an absolutely integrable function $k$, i.e. $\int_{\mathbb{R}}|k(y)|\mathrm dy<\infty$, and fixed $\bm{u}$:
\begin{align}\label{lem:ex}
\mathbb{E}[k(Y)\partial_{u_i}g_{\bm{X}}(\bm{u})]&=\mathbb{E}[\mathbb{E}[k(Y)\partial_{u_i}g_{\bm{X}}(\bm{u})~|~\bm{X}_{-1}]]\notag\\
&=\mathbb{E}\biggl[\int_{-\infty}^{+\infty}\partial_{u_i} g_{(x_1,\bm{X}_{-1})}(\bm{u})k(g_{(x_1,\bm{X}_{-1})}(\bm{u}))\phi(x_1|\bm{X}_{-1})dx_1\biggl].
\end{align} 
We now apply a change of variable:
\begin{equation}\label{eqn:chova}
x_1=l_{(y,\bm{x}_{-1})}(\bm{u})\iff y=g_{(x_1,\bm{x}_{-1})}(\bm{u}).
\end{equation}
For any $\bm{u}\in U$ and $\bm{x}_{-1}\in\mathbb{R}^{n-1}$, we can write (\ref{eqn:chova}) as:
\begin{equation}\label{eqn:linv}
x_1=l_{(g_{(x_1,\bm{x}_{-1})}(\bm{u}),\bm{x}_{-1})}(\bm{u})\iff y=g_{(l_{(y,\bm{x}_{-1})}(\bm{u}),\bm{x}_{-1})}(\bm{u}).
\end{equation}
so that:
\begin{equation}\label{eqn:dchova}
\frac{\mathrm dx_1}{\mathrm dy}=\partial_y l_{(y,\bm{x}_{-1})}(\bm{u})|_{y=g_{(x_1,\bm{x}_{-1})}(\bm{u})}=(\partial_{x_1} g_{(x_1,\bm{x}_{-1})}(\bm{u})|_{x_1=l_{(y,\bm{x}_{-1})}(\bm{u})})^{-1}\,,
\end{equation}
where we used the representation of derivatives of inverse functions. 
Next, we compute the partial derivative of the equation $y=g_{(l_{(y,\bm{x}_{-1})}(\bm{u}),\bm{x}_{-1})}(\bm{u})$ in (\ref{eqn:linv}) with respect to $u_i$, $i=1,\dots,n$, and note that the derivatives of the LHS are zero, i.e. $\partial y/\partial u_i=0$ for $i=1,\dots,n$. For the RHS, we have:
\begin{equation*}
\partial_{u_i}g_{(l_{(y,\bm{x}_{-1})}(\bm{u}),\bm{x}_{-1})}(\bm{u})=\partial_{u_i}l_{(y,\bm{x}_{-1})}(\bm{u})\partial_{x_1} g_{(x_1,\bm{x}_{-1})}(\bm{u})|_{x_1=l_{(y,\bm{x}_{-1})}(\bm{u})}+\partial_{u_i}g_{(x_1,\bm{x}_{-1})}(\bm{u})|_{x_1=l_{(y,\bm{x}_{-1})}(\bm{u})}.
\end{equation*}
From this we deduce that:
\begin{equation}\label{eqn:chova2}
\partial_{u_i}l_{(y,\bm{x}_{-1})}(\bm{u})\partial_{x_1} g_{(x_1,\bm{x}_{-1})}(\bm{u})|_{x_1=l_{(y,\bm{x}_{-1})}(\bm{u})}=-\partial_{u_i}g_{(x_1,\bm{x}_{-1})}(\bm{u})|_{x_1=l_{(y,\bm{x}_{-1})}(\bm{u})}\,.
\end{equation}
Using (\ref{eqn:dchova}) and (\ref{eqn:chova2}), our expectation in (\ref{lem:ex}) now becomes:  
\begin{align}
\mathbb{E}[k(Y)\partial_{u_i}g_{\bm{X}}(\bm{u})]&=\mathbb{E}\biggl[-\int_{-\infty}^{+\infty}k(y)\phi(l_{(y,\bm{X}_{-1})}(\bm{u})|\bm{X}_{-1})\partial_{u_i}l_{(y,\bm{X}_{-1})}(\bm{u})dy\biggl]\notag\\
&=\mathbb{E}\biggl[-\int_{-\infty}^{+\infty}k(y)\frac{\partial_{u_i}l_{(y,\bm{X}_{-1})}(\bm{u})\phi(l_{(y,\bm{X}_{-1})}(\bm{u})|\bm{X}_{-1})}{f_Y(y)}f_Y(y)dy\biggl]\label{eqn:intsw1}\\
&=-\int_{-\infty}^{+\infty}k(y)\frac{\mathbb{E}[\partial_{u_i}l_{(y,\bm{X}_{-1})}(\bm{u})\phi(l_{(y,\bm{X}_{-1})}(\bm{u})|\bm{X}_{-1})]}{f_Y(y)}f_Y(y)dy\label{eqn:intsw2}\\
&=\mathbb{E}[k(Y)q(Y)]\notag,
\end{align}
where:
\begin{equation*}
q(y)=-\frac{\mathbb{E}[\partial_{u_i}l_{(y,\bm{X}_{-1})}(\bm{u})\phi(l_{(y,\bm{X}_{-1})}(\bm{u})|\bm{X}_{-1})]}{f_Y(y)}.
\end{equation*}
Notice that we have switched the order of integration and expectation to move from (\ref{eqn:intsw1}) to (\ref{eqn:intsw2}). This can be justified by considering the following integral on the product space $\mathbb{R}\times\mathbb{R}^{n-1}$:
\begin{equation}\label{eqn:prodint}
I_{k}:=\int_{\mathbb{R}\times\mathbb{R}^{n-1}}|k(y)\phi(l_{(y,\bm{x}_{-1})}(\bm{u})|\bm{x}_{-1})\partial_{u_i}l_{(y,\bm{x}_{-1})}(\bm{u})|\,dy\, dF_{\bm{X}_{-1}}(\bm{x}_{-1}).
\end{equation}
If the integral (\ref{eqn:prodint}) is finite then changing the order of integrals in (\ref{eqn:intsw2}) is justified by Fubini's theorem. Observe that since the integrand is non-negative, we can apply Tonelli's theorem to (\ref{eqn:prodint}):
\begin{align}
I_k&=\int_{\mathbb{R}}\int_{\mathbb{R}^{n-1}}|k(y)\phi(l_{(y,\bm{x}_{-1})}(\bm{u})|\bm{x}_{-1})\partial_{u_i}l_{(y,\bm{x}_{-1})}(\bm{u})|\,dF_{\bm{X}_{-1}}(\bm{x}_{-1})\,dy\notag
\\
&=
\int_{\mathbb{R}}\int_{\mathbb{R}^{n-1}}|k(y)||\phi(l_{(y,\bm{x}_{-1})}(\bm{u})|\bm{x}_{-1})\partial_{u_i}l_{(y,\bm{x}_{-1})}(\bm{u})|\,dF_{\bm{X}_{-1}}(\bm{x}_{-1})\,dy\notag
\\
&=
\int_{\mathbb{R}}|k(y)|\int_{\mathbb{R}^{n-1}}|\phi(l_{(y,\bm{x}_{-1})}(\bm{u})|\bm{x}_{-1})\partial_{u_i}l_{(y,\bm{x}_{-1})}(\bm{u})|\,dF_{\bm{X}_{-1}}(\bm{x}_{-1})\,dy\notag
\\
&=
\int_{\mathbb{R}}|k(y)|\biggl(\mathbb{E}\biggl[\max\biggl(\frac{\partial G}{\partial u_i}(y,\bm{u},\bm{X}_{-1}),-\frac{\partial G}{\partial u_i}(y,\bm{u},\bm{X}_{-1})\biggl)\biggl]\biggl)\,dy,\label{eqn:Ik}
\end{align}
where, as in the proof of Lemma \ref{lem:F}, we used that:
\begin{equation*}
\frac{\partial G}{\partial u_i}(y,\bm{u},\bm{x}_{-1})=\partial_{u_i}l_{(y,\bm{x}_{-1})}(\bm{u})\phi(l_{(y,\bm{x}_{-1})}(\bm{u})|\bm{x}_{-1}).
\end{equation*}
By Assumption \ref{asmp:t} (\ref{asmp:fint}), the expectation in the integrand of (\ref{eqn:Ik}) is finite and since $k$ is absolutely integrable, we conclude that $I_k<\infty$. Using the explicit form of $f_Y$ from Remark \ref{rem:dens}, we conclude that:
\begin{equation*}
\mathbb{E}[\partial_{u_i}g_{\bm{X}}(\bm{u})~|~Y=y]=-\frac{\mathbb{E}[\partial_{u_i}l_{(y,\bm{X}_{-1})}(\bm{u})\phi(l_{(y,\bm{X}_{-1})}(\bm{u})|\bm{X}_{-1})]}{\mathbb{E}[\partial_y l_{(y,\bm{X}_{-1})}(\bm{u})\phi( l_{(y,\bm{X}_{-1})}(\bm{u})|\bm{X}_{-1})]}.
\end{equation*}
\end{proof}

\begin{remark}\label{rem:hongps}
Using Lemmas \ref{lem:F} and \ref{lem:cexp} and Remark \ref{rem:dens}, one can write the derivative of the portfolio with respect to its composition as:
\begin{align}
\frac{\partial F}{\partial u_i}(y,\bm{u})&=\mathbb{E}[\partial_{u_i} l_{(y,\bm{X}_{-1})}(\bm{u})\phi(l_{(y,\bm{X}_{-1})}(\bm{u})|\bm{X}_{-1})]\notag\\
&=-\mathbb{E}[\partial_y l_{(y,\bm{X}_{-1})}(\bm{u})\phi( l_{(y,\bm{X}_{-1})}(\bm{u})|\bm{X}_{-1})]\mathbb{E}[\partial_{u_i}g_{\bm{X}}(\bm{u})~|~Y=y]\notag\\
&=-f_Y(y)\mathbb{E}[\partial_{u_i}g_{\bm{X}}(\bm{u})~|~g_{\bm{X}}(\bm{u})=y].
\end{align}
\end{remark}

Using Assumption \ref{asmp:h}, we demonstrate Proposition \ref{propo:rhocont} which proves partial differentiability of lambda quantiles without assuming $g_{\bm{X}}(\bm{u})$ is $\mathbb{P}$-a.s. strictly increasing. Instead, we assume that the portfolio has a continuous density in a neighbourhood of the lambda quantile. Assumption \ref{asmp:h} and Proposition \ref{propo:rhocont} are then used to prove Theorem \ref{thm:tlqrc} with condition \ref{thmasmp:h}.

\begin{proposition}\label{propo:rhocont}
Suppose Assumption \ref{asmp:h} \eqref{asmp:h2a}, \eqref{asmp:h2b} and \eqref{asmp:grad} are satisfied and $\Lambda$ is continuously differentiable in a neighbourhood of $-\rho_{\Lambda}(\bm{u})$. Then, $\rho_{\Lambda}$ is continuously partially differentiable in $U$ with derivatives:
\begin{equation*}
\frac{\partial \rho_\Lambda}{\partial u_i}(\bm{u})=\biggl(\frac{\partial H}{\partial y}(-\rho_{\Lambda}(\bm{u}),\bm{u})\biggl)^{-1}\frac{\partial H}{\partial u_i}(-\rho_{\Lambda}(\bm{u}),\bm{u}),
\end{equation*}
where $H(y,\bm{u}):=F(y,\bm{u})-\Lambda(y)$.
\end{proposition}
\begin{proof}
Fix $\bm{u}\in U$. Then, $g_{\bm{X}}(\bm{u})$ is $\mathbb{P}$-a.s. a continuous random variable in a neighbourhood of $-\rho_{\Lambda}(\bm{u})$. Therefore, it holds that:
\begin{equation*}
F(-\rho_{\Lambda}(\bm{u}),\bm{u})=\Lambda(-\rho_{\Lambda}(\bm{u})).
\end{equation*}
Then, $y=-\rho_{\Lambda}(\bm{u})$ is a solution of $H(y,\bm{u})=0$ for all $\bm{u}\in U$, i.e. $H(-\rho_{\Lambda}(\bm{u}), \bm{u}) = 0$ for all $\bm{u}\in U$. Note that $H$ is continuously partially differentiable in $y$ and $u_i$, $i=1,\dots,n$, since by assumption, both $f_Y$ and $\Lambda'$ are continuous in the same neighbourhood of $-\rho_{\Lambda}(\bm{u})$. Also, observe that:
\begin{equation*}
\frac{\partial H}{\partial y}(y,\bm{u})\biggl|_{y=-\rho_{\Lambda}(\bm{u})}=f_Y(-\rho_{\Lambda}(\bm{u}))-\Lambda'(-\rho_{\Lambda}(\bm{u}))>0
\end{equation*}
by Assumption \ref{asmp:h} (\ref{asmp:grad}). Applying the implicit function theorem to $H$ and using Assumption \ref{asmp:h} (\ref{asmp:h2b}), we conclude that $-\rho_{\Lambda}$ is continuously partially differentiable in $U$ with derivatives:
\begin{equation*}
\frac{\partial (-\rho_\Lambda)}{\partial u_i}(\bm{u})=-\biggl(\frac{\partial H}{\partial y}(-\rho_{\Lambda}(\bm{u}),\bm{u})\biggl)^{-1}\frac{\partial H}{\partial u_i}(-\rho_{\Lambda}(\bm{u}),\bm{u}).
\end{equation*}
\end{proof}

We now define the portfolio density adjustment which is important for both the risk contributions and Euler decomposition of lambda quantiles. Also, we will show that the portfolio density adjustment evaluated at the point $y=-\rho_{\Lambda}(\bm{u})$ corresponds to the homogeneity degree of lambda quantiles.

\begin{definition}\label{defn:rda}
For a continuous random variable $Y\in\mathcal{X}$ and continuously differentiable lambda function, define the \emph{portfolio density adjustment of $Y$ with respect to $\Lambda$} as the function $\eta_{\Lambda,Y}:\mathbb{R}\rightarrow\mathbb{R}\cup\{+\infty\}$ given by
\begin{equation}\label{rdp}
\eta_{\Lambda,Y}(y):=\frac{f_Y(y)}{f_Y(y)-\Lambda'(y)},
\end{equation}
where we use the convention that $\frac{1}{0} = + \infty$.
\end{definition}
\begin{remark}
Observe that $\eta_{\Lambda,Y}(y)=1$ at a given $y$ if, and only if, $\Lambda'(y)=0$. Also, for a fixed lambda function, the portfolio density adjustment $\eta_{\Lambda,Y}$ is law invariant with respect to the random variable $Y$, that is, for random variables $Y_1,Y_2\in\mathcal{X}$ that are equal in distribution, i.e. $Y_1\,{\buildrel d \over =}\,Y_2$, it holds  $\eta_{\Lambda,Y_1}(y)=\eta_{\Lambda,Y_2}(y)$ for all $y\in\mathbb{R}$.
\end{remark}

The following theorem states the conditions under which lambda quantiles are continuously partially differentiable in the space of portfolio compositions and provides closed form formulae of lambda quantile risk contributions. We prove Theorem \ref{thm:tlqrc} using two different approaches, which correspond to the use of Assumption \ref{asmp:t} and Assumption \ref{asmp:h}. Also, note that the assumptions for the lambda function are different for each approach.

\begin{theorem}\label{thm:tlqrc}
Suppose either:
\begin{enumerate}
\item $\Lambda$ is continuously differentiable on $\mathbb{R}$ and Assumption \ref{asmp:t} is satisfied,\label{thmasmp:t}\\
or
\item $\Lambda$ is continuously differentiable in a neighbourhood of $y=-\rho_{\Lambda}(\bm{u})$ and Assumption \ref{asmp:h} is satisfied.\label{thmasmp:h}
\end{enumerate}
Then, $\rho_{\Lambda}$ is continuously partially differentiable in $U$ with partial derivatives:
\begin{equation}\label{eqn:lqrc}
\frac{\partial \rho_{\Lambda}}{\partial u_i}(\bm{u})=-\eta_{\Lambda,Y}(-\rho_{\Lambda}(\bm{u}))\mathbb{E}[\partial_{u_i} g_{\bm{X}}(\bm{u})~|~Y=-\rho_{\Lambda}(\bm{u})],
\end{equation}
for $i=1,\dots,n$.
\end{theorem}

\begin{proof}[Proof of Theorem \ref{thm:tlqrc}]
\textit{Proof using condition \ref{thmasmp:t}}.
By Lemma \ref{lem:F} and Remark \ref{rem:dens}, $Y=g_{\bm{X}}(\bm{u})$ is a continuous random variable with a continuous probability density function $f_Y$ and the partial derivatives $\partial_{u_i}F(y,\bm{u})$ are continuous in $y$ and $\bm{u}$ for $i=1,\dots,n$. Furthermore, with Assumption \ref{asmp:t} \eqref{asmp:elq}, we invoke Proposition \ref{propo:rhocont} to deduce that $\rho_{\Lambda}$ is continuously partially differentiable in $U$ with derivatives:
\begin{equation*}
\frac{\partial \rho_\Lambda}{\partial u_i}(\bm{u})=\frac{\partial_{u_i}F(-\rho_{\Lambda}(\bm{u}),\bm{u})}{f_Y(-\rho_{\Lambda}(\bm{u}))-\Lambda'(-\rho_{\Lambda}(\bm{u}))},
\end{equation*}
for $i=1,\dots,n$. By Remark \ref{rem:hongps}, we note that:
\begin{equation*}
\partial_{u_i}F(-\rho_{\Lambda}(\bm{u}),\bm{u})=-f_Y(-\rho_{\Lambda}(\bm{u}))\mathbb{E}[\partial_{u_i} g_{\bm{X}}(\bm{u})~|~g_{\bm{X}}(\bm{u})=-\rho_{\Lambda}(\bm{u})],
\end{equation*}
which concludes the proof using condition \ref{thmasmp:t}, since:
\begin{equation*}
\frac{\partial \rho_\Lambda}{\partial u_i}(\bm{u})=-\frac{f_Y(-\rho_{\Lambda}(\bm{u}))}{f_Y(-\rho_{\Lambda}(\bm{u}))-\Lambda'(-\rho_{\Lambda}(\bm{u}))}\mathbb{E}[\partial_{u_i} g_{\bm{X}}(\bm{u})~|~g_{\bm{X}}(\bm{u})=-\rho_{\Lambda}(\bm{u})].
\end{equation*}

\textit{Proof using condition \ref{thmasmp:h}}.
By Proposition \ref{propo:rhocont}, $\rho_{\Lambda}$ is continuously partially differentiable in $U$ with partial derivatives:
\begin{equation*}
\frac{\partial \rho_\Lambda}{\partial u_i}(\bm{u})=\biggl(\frac{\partial H}{\partial y}(-\rho_{\Lambda}(\bm{u}),\bm{u})\biggl)^{-1}\frac{\partial H}{\partial u_i}(-\rho_{\Lambda}(\bm{u}),\bm{u}),
\end{equation*}
where $H(y,\bm{u}):=F(y,\bm{u})-\Lambda(y)$. By Theorem 1 of \cite{h09}, we have:
\begin{equation*}
\partial_{u_i} F(y,\bm{u})=-f_Y(y)\mathbb{E}[\partial_{u_i}g_{\bm{X}}(\bm{u})~|~g_{\bm{X}}(\bm{u})=y],
\end{equation*}
for $i=1,\dots,n$, which are continuous in a neighbourhood of $y=-\rho_{\Lambda}(\bm{u})$ by Assumption \ref{asmp:h} (\ref{asmp:h2b}). Furthermore, observe that:
\begin{equation*}
\frac{\partial H}{\partial y}(y,\bm{u})=f_Y(y)-\Lambda'(y),
\end{equation*}
which, again, is continuous in a neighbourhood of $y=-\rho_{\Lambda}(\bm{u})$ by Assumption \ref{asmp:h} (\ref{asmp:h2a}) and by the assumption that $\Lambda$ is continuously differentiable in a neighbourhood of $y=-\rho_{\Lambda}(\bm{u})$. We conclude that the continuous partial derivatives of $\rho_{\Lambda}$ are given by:
\begin{equation*}
\frac{\partial \rho_\Lambda}{\partial u_i}(\bm{u})=-\frac{f_Y(-\rho_{\Lambda}(\bm{u}))}{f_Y(-\rho_{\Lambda}(\bm{u}))-\Lambda'(-\rho_{\Lambda}(\bm{u}))}\mathbb{E}[\partial_{u_i}g_{\bm{X}}(\bm{u})~|~g_{\bm{X}}=-\rho_{\Lambda}(\bm{u})],
\end{equation*}
for $i=1,\dots,n$.
\end{proof}

Theorem \ref{thm:tlqrc} with condition \ref{thmasmp:h} is a generalisation of the quantile sensitivity of VaR derived in Theorem 2 of \cite{h09} to the class of lambda quantiles.

For the special case of $VaR_{\lambda}$, we observe that the portfolio density adjustment is equal to one, i.e. $\eta_{\lambda,Y}(y)=1$ for all $y\in\mathbb{R}$, which leads to the following result.

\begin{corollary}\label{cor:varrc}
Suppose $\Lambda(x)=\lambda\in(0,1)$ for all $x\in\mathbb{R}$. If Assumption \ref{asmp:t} or Assumption \ref{asmp:h} is satisfied, then $\rho_{\lambda}\equiv VaR_{\lambda}$ is continuously partially differentiable in $U$ with partial derivatives:
\begin{equation}\label{eqn:varrc}
\frac{\partial VaR_{\lambda}}{\partial u_i}(\bm{u})=-\mathbb{E}[\partial_{u_i} g_{\bm{X}}(\bm{u})~|~Y=-VaR_{\lambda}(\bm{u})],
\end{equation}
for $i=1,\dots,n$.
\end{corollary}
The Proof of Corollary \ref{cor:varrc} follows straightforwardly from Theorem \ref{thm:tlqrc}, and is thus omitted. Corollary \ref{cor:varrc} with Assumption \ref{asmp:t} generalises Lemma 5.3 in \cite{t99} to generic portfolios $g_{\bm{X}}(\bm{u})$. Furthermore, even though (\ref{eqn:varrc}) with Assumption \ref{asmp:t} is of the same form as the partial derivative given in Theorem 2 of \cite{h09} (except that here $\bm{u}$ is multivariate as opposed to one-dimensional), the assumptions used to obtain these results differ. In \cite{h09}, the simulation output is assumed to be a continuous random variable. Corollary \ref{cor:varrc} with Assumption \ref{asmp:t}, does not require this assumption, we do, however, assume that at least one of the $X_i$ has a continuous density.

The following example shows that the risk contributions of $VaR_{\lambda}$ in \cite{t99}, who considers linear portfolios, are a special case of those of the lambda quantiles.

\begin{example}\label{exmp:lp}
For the linear portfolio operator given in Example \ref{exmp:ops} (\ref{op:lin}), we fix the random vector $\bm{X}$ to obtain the portfolio:
\begin{equation*}
g_{\bm{X}}(\bm{u}):=u_1X_1+u_2X_2.
\end{equation*}
Then the lambda quantiles' risk contribution of asset $i$ to the portfolio is given by:
\begin{equation*}
\frac{\partial \rho_{\Lambda}}{\partial u_i}(\bm{u})=-\eta_{\Lambda,Y}(-\rho_{\Lambda}(\bm{u}))\mathbb{E}[X_i|u_1X_1+u_2X_2=-\rho_{\Lambda}(\bm{u})]\,,
\end{equation*}
where $Y = u_1X_1+u_2X_2$ as given in Example \ref{exmp:ops} (\ref{op:lin}).
If $\Lambda(x)=\lambda\in(0,1)$ is a constant, then we retrieve partial derivatives of $VaR_{\lambda}$ as obtained in \cite{t99}, \cite{gls00}, and \cite{h03}:
\begin{equation*}
\frac{\partial VaR_{\lambda}}{\partial u_i}(\bm{u})=-\mathbb{E}[X_i|u_1X_1+u_2X_2=-VaR_{\lambda}(\bm{u})].
\end{equation*}
\end{example}

\begin{example}\label{exmp:non-hom}
Consider the portfolio operator in Example \ref{exmp:ops} (\ref{op:lq}) with $\tau=1$ such that $g_{\bm{X}}$ is given by:
\begin{equation*}
g_{\bm{X}}(\bm{u})=u_1X_1+u_2X_2-\rho_{\Lambda}(\bm{u};Y),
\end{equation*}
where $Y=u_1X_1+u_2X_2$. Then, the lambda quantile admits the representation
\begin{equation*}
\rho_{\Lambda}(\bm{u};g_{\bm{X}})=\rho_{\Gamma}(\bm{u};Y)+\rho_{\Lambda}(\bm{u};Y),
\end{equation*}
where $\Gamma(z):=\Lambda(z-\rho_{\Lambda}(\bm{u};Y))$ for all $z\in\mathbb{R}$. Thus, the risk contributions of the lambda quantile for the portfolio $g_{\bm{X}}(\bm{u})$, for $i = 1, \ldots, n$, become 
\begin{align*}
\frac{\partial \rho_{\Lambda}}{\partial u_i}(\bm{u};g_{\bm{X}})
&= \frac{\partial \rho_{\Gamma}}{\partial u_i}(\bm{u};Y) + \frac{\partial \rho_{\Lambda}}{\partial u_i}(\bm{u};Y)\\
&=-\eta_{\Gamma, Y}(-\rho_{\Gamma}(\bm{u};Y))\mathbb{E}[X_i~|~Y=-\rho_{\Gamma}(\bm{u};Y)]
\\
& \quad -\eta_{\Lambda,Y}(-\rho_{\Lambda}(\bm{u};Y))\mathbb{E}[X_i~|~Y=-\rho_{\Lambda}(\bm{u};Y)].
\end{align*}
\end{example}

So far we proved in Theorem \ref{thm:tlqrc} that, under smoothness assumptions, lambda quantiles are continuously partially differentiable in $U$. Next, we consider differentiability in subsets of $U$. This is important for situations when portfolio selection is restricted to specific classes of compositions, or in other words, to subsets of $U$. In the following proposition, we use Assumption \ref{asmp:t} to prove that lambda quantiles are continuously partially differentiable in subsets of $U$. This result allows for flexibility in the choice of lambda function of lambda quantiles. Recall that in Theorem \ref{thm:tlqrc} with condition \ref{thmasmp:t}, $\Lambda$ was assumed to be continuously differentiable in $\mathbb{R}$. In Proposition \ref{propo:tlqrci}, we only require continuous differentiability of the lambda function within an interval, thus generalising to lambda functions that may be discontinuous on $\mathbb{R}$.

Consider a subset $V\subset U$ such that for all $\bm{v}\in V$, the smallest intersection point of $F$ and $\Lambda$ lies in the interval $(\alpha,\beta)\subset\mathbb{R}$, i.e. $-\rho_{\Lambda}(\bm{v})\in (\alpha,\beta)$ for all $\bm{v}\in V$. The following result provides the necessary conditions to ensure lambda quantiles are continuously partially differentiable in $V$, and hence allow us to calculate risk contributions of lambda quantiles in $V$.

\begin{proposition}\label{propo:tlqrci}
Assume that:
\begin{enumerate}[(i)]
\item $\Lambda$ is continuously differentiable in the interval $(\alpha,\beta)\subset\mathbb{R}$;
\item Assumption \ref{asmp:t} is satisfied;
\item $-\rho_{\Lambda}(\bm{v})\in (\alpha,\beta)$ for all $\bm{v}\in V\subset U$.
\end{enumerate}
Then, $\rho_{\Lambda}$ is continuously partially differentiable in $V$, where the partial derivatives are given by:
\begin{equation}\label{eqn:lqrcV}
\frac{\partial \rho_{\Lambda}}{\partial v_i}(\bm{v})=-\eta_{\Lambda,Y}(-\rho_{\Lambda}(\bm{v}))\mathbb{E}[\partial_{v_i} g_{\bm{X}}(\bm{v})~|~Y=-\rho_{\Lambda}(\bm{v})],
\end{equation}
for $\bm{v} \in V$ and $i=1,\dots,n$.
\end{proposition}
\begin{proof}
The proof follows the same approach as that of Theorem \ref{thm:tlqrc} with condition \ref{thmasmp:t}. The major difference is that we calculate partial derivatives with respect to $y$ in the interval $(\alpha,\beta)$ to ensure $\Lambda'(y)$ exists and is well defined.
We further point out that for fixed $\bm{v}\in V$, $-\rho_{\Lambda}(\bm{v})$ is the smallest intersection point of $F$ and $\Lambda$ on $(\alpha,\beta)$, since they're both continuous on this interval, i.e. we have:
\begin{equation*}
F(-\rho_{\Lambda}(\bm{v}),\bm{v})=\Lambda(-\rho_{\Lambda}(\bm{v})),
\end{equation*}
for all $\bm{v}\in V$.
\end{proof}

\section{Euler decomposition and the generalised Euler allocation rule}\label{sec:ed}

In this section, we aggregate the risk contributions of lambda quantiles to prove a relationship known as the \emph{Euler decomposition} for lambda quantiles. The Euler decomposition is, for homogeneous risk measures in $U$ with homogeneity degree 1, the property that the risk measure, scaled by its homogeneity degree, can be written as a sum of its partial derivatives scaled by the number of assets. We show that the homogeneity degree of lambda quantiles is determined by the portfolio composition, the density function of the portfolio, and the gradient of the lambda function, both evaluated at the lambda quantile. This implies that lambda quantile homogeneity degree is not constant over choices of portfolio compositions or lambda functions. Furthermore, the homogeneity degree varies across different distributions of the portfolio. This is in contrast to other risk measures, such as VaR, where the homogeneity degree is constant.

In risk measure theory, the property of homogeneity is typically studied for risk measures defined on the space of random variables. A risk measure defined on random variables is positive homogeneous (of degree 1), if the risk of an asset scales linearly, e.g. doubling the asset's units doubles the position's risk. The positive homogeneity (of degree 1) property of risk measures forms part of the definition of coherent risk measures, introduced in the seminal paper by \cite{a99}. However, the property of a risk measure having a homogeneity degree of 1 has been questioned in \cite{fs02}. They argue that large multiples of a position may introduce additional liquidity risk and, therefore, the position's risk and size may not increase linearly.

In this paper, we study the homogeneity property of lambda quantiles on the set $U$ and explore the relationship between asset units and portfolio risk for lambda quantiles. Therefore, our treatment of the homogeneity property should not be confused with homogeneity of risk measures defined on set of random variables $\mathcal{X}$.

The Euler decomposition of a positive homogeneous (of degree 1) risk measure, defined on the set of random variables, is known as the \emph{Euler allocation rule} \citep{p99,d01,t07}, which is one of the most well established allocation methods in risk measure theory. This allocation rule assigns economic capital to assets using directional derivatives (in the direction of the asset) of positive homogeneous risk measures. Furthermore, Euler allocation rule is used for portfolios with linear risk aggregation or linear portfolio operators in \cite{t07} and \cite{ts09}.

Our treatment considers a more general setup, where we consider a generic portfolio operator and lambda quantiles, that are risk measures with non-constant homogeneity degree.

\begin{definition}\label{def:hom}
Let $\alpha \colon U \to \mathbb{R}$ be a function. 
We call a function $r:U\rightarrow\mathbb{R}$ positively homogeneous of degree $\alpha(\bm{u})$, if for all $\bm{u} \in U$ and $t>0$ such that $t \bm{u} \in U$, it holds that
\begin{equation}
    r(t\bm{u}) = t^{\alpha(\bm{u})} r(\bm{u})\,.
\end{equation}
\end{definition}
A function that is positively homogeneous of degree $\alpha(\bm{u})$ satisfies an Euler-like theorem. Indeed, a differentiable function $r:U\rightarrow\mathbb{R}$ is positively homogeneous of degree $\alpha(\bm{u})$ if, and only if  for all $\bm{u} \in U$ it holds that
\begin{equation}\label{eq: Euler}
    \sum_{i = 1}^n u_i \,\frac{\partial r}{\partial u_i}  (\bm{u}) = \alpha(\bm{u}) r(\bm{u})\,.
\end{equation}
Note that if $\alpha(\bm{u}) = \alpha \in \mathbb R$ is a constant, then we recover the usual definition of positively homogeneous functions. 

\begin{definition}\label{defn:esh}
Let $\tau \colon U \to \mathbb{R}$ and $E\in\mathcal{F}$ be an event. An operator $g:U\times\mathcal{X}^n\rightarrow\mathcal{X}$ is said to be \emph{$\mathbb{P}$-almost surely positively homogeneous of degree $\tau(\bm{u})$} in $U$ and in the event $E$ if for all $\bm{X}\in\mathcal{X}^n$, and all $\bm{u}\in U$ and $t>0$ with $t\bm{u} \in U$, we have:
\begin{equation}
\mathbb{P}\left(\{\omega\in E:g[t\bm{u},\bm{X}](\omega)=t^{\tau(\bm{u})}g[\bm{u},\bm{X}](\omega)\}\right)=1\,.
\end{equation}
If $E=\Omega$, we say $g$ is \emph{$\mathbb{P}$-a.s. positively $\tau$-homogeneous for a function $\tau$ in $U$}.
\end{definition}

Observe that if $g$ is $\mathbb{P}$-a.s. positively $\tau$-homogeneous for a function $\tau$ in $U$ and in the event $E\in\mathcal{F}$, then $g_{\bm{X}}$ is also $\mathbb{P}$-a.s. positively $\tau$-homogeneous for a function $\tau$ in $U$ and in the event $E$ for all $\bm{X}\in\mathcal{X}^n$.
Moreover, any portfolio operator that is linear in $U$ is $\mathbb{P}$-a.s. 1-homogeneous in $U$. For a non-linear operator $g$, however, the homogeneity property may not hold for all $\omega\in\Omega$.
In contrast, Definition \ref{defn:esh} applies to functions which map onto random variables, where there may exist outcomes $\omega$ for which the operator $g$ is not homogeneous in $U$. Therefore, $\mathbb{P}$-a.s. homogeneity is especially appealing to non-linear portfolio operators.

\begin{theorem}\label{thm:ed}
Suppose $\Lambda$ is continuously differentiable on $\mathbb{R}$, $g_{\bm{X}}$ is $\mathbb{P}$-a.s. positively $\tau$-homogeneous in $U$, for some function $\tau \colon U\to \mathbb{R}$, and $g_{\bm{x}}(\bm{u})$ is differentiable in $U$, for any fixed $\bm{x}\in\mathbb{R}^n$. If either Assumption \ref{asmp:t} or Assumption \ref{asmp:h} is satisfied, then for all $\bm{u}\in U$, $\rho_{\Lambda}$ satisfies:
\begin{equation}\label{lqed}
\tau(\bm{u})\;\eta_{\Lambda,Y}(-\rho_{\Lambda}(\bm{u}))\;\rho_{\Lambda}(\bm{u})=\sum_{i=1}^nu_i\frac{\partial \rho_{\Lambda}}{\partial u_i}(\bm{u}).
\end{equation}
\end{theorem}
\begin{proof}
If $g_{\bm{X}}$ is $\mathbb{P}$-a.s. positively $\tau$-homogeneous in $U$ for a function $\tau$ and $g_{\bm{x}}(\bm{u})$ is differentiable in $\bm{u}$, for all $\bm{u} \in U$ and fixed $\bm{x}\in\mathbb{R}^n$, then by Equation \eqref{eq: Euler} it holds for almost all $\omega\in\Omega$ that
\begin{equation}\label{eqn:gxoed1}
\tau(\bm{u})\;g_{\bm{X}(\omega)}(\bm{u})=\sum_{i=1}^n u_i\partial_{u_i} g_{\bm{X}(\omega)}(\bm{u}),
\end{equation}
for all $\bm{u}\in U$.
We note that (\ref{eqn:gxoed1}) is equivalent to
\begin{equation}\label{eqn:gxoed3}
\tau (\bm{u}) \;g_{\bm{X}}(\bm{u})=\sum_{i=1}^n u_i\partial_{u_i} g_{\bm{X}}(\bm{u})\quad\mathbb{P}\text{-a.s.},
\end{equation}
for all $\bm{u}\in U$. In (\ref{eqn:gxoed3}), we have equivalence of two random variables in a $\mathbb{P}$-a.s. sense, thus they have $\mathbb{P}$-a.s. equal conditional expectations, that is
\begin{equation}\label{eqn:ceas}
\mathbb{E}[\tau(\bm{u}) g_{\bm{X}}(\bm{u})~|~Y]=\mathbb{E}\biggl[\sum_{i=1}^n u_i\partial_{u_i} g_{\bm{X}}(\bm{u})~|~Y\biggl]\quad\mathbb{P}\text{-a.s.},
\end{equation}
where $Y:= g_{\bm{X}}(\bm{u})\in\mathcal{X}$. Note that under Assumption \ref{asmp:t}, we apply Theorem \ref{thm:tlqrc} with condition \ref{thmasmp:t}, and under Assumption \ref{asmp:h}, we apply Theorem \ref{thm:tlqrc} with condition \ref{thmasmp:h}, to obtain partial derivatives of $\rho_{\Lambda}$. Recall expression (\ref{eqn:lqrc}) from Theorem \ref{thm:tlqrc}:
\begin{equation*}
\frac{\partial \rho_{\Lambda}}{\partial u_i}(\bm{u})=-\eta_{\Lambda,Y}(-\rho_{\Lambda}(\bm{u}))\mathbb{E}[\partial_{u_i} g_{\bm{X}}(\bm{u})~|~Y=-\rho_{\Lambda}(\bm{u})],
\end{equation*}
for $i=1,\dots,n$. Note that the conditioning event $Y=-\rho_{\Lambda}(\bm{u})$ in the expectation is the same for all $i$. Therefore, summing the risk contributions scaled by the number of assets over $i$ and using (\ref{eqn:ceas}), we obtain
\begin{align*}
\sum_{i=1}^n u_i\frac{\partial \rho_{\Lambda}}{\partial u_i}(\bm{u})&=-\eta_{\Lambda,Y}(-\rho_{\Lambda}(\bm{u}))\sum_{i=1}^n u_i\mathbb{E}[\partial_{u_i} g_{\bm{X}}(\bm{u})~|~Y=-\rho_{\Lambda}(\bm{u})]\nonumber\\
&=-\eta_{\Lambda,Y}(-\rho_{\Lambda}(\bm{u}))\mathbb{E}\biggl[\sum_{i=1}^n u_i\partial_{u_i} g_{\bm{X}}(\bm{u})~|~Y=-\rho_{\Lambda}(\bm{u})\biggl]\\
&=-\eta_{\Lambda,Y}(-\rho_{\Lambda}(\bm{u}))\mathbb{E}[\tau(\bm{u})\; g_{\bm{X}}(\bm{u})~|~Y=-\rho_{\Lambda}(\bm{u})]\\
&=\tau(\bm{u})\;\eta_{\Lambda,Y}(-\rho_{\Lambda}(\bm{u}))\;\rho_{\Lambda}(\bm{u})\,,
\end{align*}
where the last equation holds by continuity of $F_Y$.
\end{proof}

From Theorem \ref{thm:ed} we conclude that lambda quantiles are positively homogeneous of degree given in the next proposition.
\begin{proposition}\label{prop:homo-lambda}
Let $g_{\bm{X}}$ be $\mathbb{P}$-a.s. positively $\tau$-homogeneous in $U$, for a function $\tau \colon U \to \mathbb R$. Then,
$\rho_{\Lambda}$ applied to $g_{\bm{X}}$ is homogeneous in $U$ of degree $\tau(\bm{u})\;\eta_{\Lambda,Y}(-\rho_{\Lambda}(\bm{u}))$. That is, the lambda quantile $\rho_\Lambda$ is $\gamma$-homogeneous in $U$ for the function $\gamma\colon U \to \mathbb{R}$, defined by $\gamma(\bm{u}) := \tau(\bm{u})\;\eta_{\Lambda,Y}(-\rho_{\Lambda}(\bm{u}))$.
\end{proposition}

Theorem \ref{thm:ed} and Proposition \ref{prop:homo-lambda} has several interesting implications. To begin with, the homogeneity degree of a lambda quantile is $\tau(\bm{u}) \, \eta_{\Lambda,Y}(-\rho_{\Lambda}(\bm{u}))$, a composition of the homogeneity degree of the portfolio $\tau(\bm{u})$ and the portfolio density adjustment $\eta_{\Lambda,Y}(-\rho_{\Lambda}(\bm{u}))$.
Thus, the homogeneity degree of $\rho_{\Lambda}$ of a linear portfolio operator (i.e. $\tau(\bm{u}) = 1$) is $\eta_{\Lambda,Y}(-\rho_{\Lambda}(\bm{u}))$. It is straightforward that the homogeneity degree of $\rho_{\Lambda}$ with a constant lambda function, $\Lambda(x)=\lambda\in(0,1)$, and for a linear portfolio operator is precisely the homogeneity degree of the $VaR_{\lambda}$ measure, that is 1. Indeed, for $\mathbb{P}$-a.s. 1-homogeneous portfolio operators, $\rho_{\Lambda}$ is 1-homogeneous if, and only if, $\Lambda'(y)=0$ for all $y\in \mathbb{R}$. Note that for $VaR_{\lambda}$ the homogeneity degree is independent of the portfolio composition $\bm{u}$. This is in contrast to a non-constant $\Lambda$ function, in which case the lambda quantile homogeneity degree may differ for each portfolio composition $\bm{u}\in U$. Moreover, for a fixed portfolio composition $\bm{u}$, the homogeneity degree of lambda quantiles may change for different choices of the $\Lambda$ function.

Next, we use Theorem \ref{thm:ed} to define a new capital allocation rule, which generalises the well-known Euler allocation. For a linear portfolio, risk contributions calculated as directional derivatives of positive homogeneous risk measures of degree 1 are known as Euler contributions \citep{t07}. Furthermore, the assignment of capital using Euler contributions is known as Euler allocation. Defining for Euler contributions is that they possess the full allocation property, i.e. the sum of the Euler contributions over all assets equals the risk measure itself. We propose a generalisation of Euler contributions which satisfies the full allocation property and that is compatible with $\gamma$-homogeneous risk measures, $\gamma\colon U \to \mathbb{R}$, and generic portfolio operators, thus applicable to lambda quantiles.
\begin{definition}\label{defn:gear}
Consider a portfolio $g_{\bm{X}}(\bm{u})$ and a risk measure $\Phi:\mathcal{X}\rightarrow\mathbb{R}$ defined on the space of random variables. Assume that the composition $\Phi \circ g_{\bm{X}} \colon U \to \mathbb{R}$ is positively $\gamma$-homogeneous for a function $\gamma \colon U \to \mathbb{R}$.
Then, the functionals $\psi_{i}^{\Phi}:\mathcal{X}\rightarrow\mathbb{R}$ defined by
\begin{equation}
\psi_{i}^\Phi(g_{\bm{X}}):=\frac{1}{\gamma(\bm{1})}\;\frac{\partial\Phi}{\partial u_i}\big(g_{\bm{X}}(\bm{u})\big)\biggl|_{\bm{u}=\bm{1}}\quad \text{for} \quad i=1,\dots,n,
\end{equation}
are called \emph{generalised Euler contributions}. Furthermore, we call the process of allocating capital to sub-portfolios using generalised Euler contributions, the \emph{generalised Euler allocation rule}.
\end{definition}

Euler contributions as defined by \citep{t07} and the Euler allocation rule \citep{p99,d01,t07} are special cases of Definition \ref{defn:gear} with $\eta=\tau =1$.

\begin{proposition}
Suppose $\Lambda$ is continuously differentiable on $\mathbb{R}$, $g_{\bm{X}}$ is $\mathbb{P}$-a.s. positively $\tau$-homogeneous in $U$, for some $\tau \colon U \to \mathbb R$, and $g_{\bm{x}}(\bm{u})$ is differentiable in $\bm{u}$, for all $\bm{u} \in U$ and fixed $\bm{x}\in\mathbb{R}^n$. Also, suppose that either Assumption \ref{asmp:t} or Assumption \ref{asmp:h} is satisfied. Then, the generalised Euler contributions of the lambda quantile are given by: 
\begin{equation}
\psi_{i}^\Lambda(g_{\bm{X}})=-\frac{1}{\tau(\bm{1})}\mathbb{E}[\partial_{u_i}g_{\bm{X}}(\bm{1})~|~g_{\bm{X}}(\bm{1})=-\rho_{\Lambda}(\bm{1})]\quad \text{for} \quad i=1,\dots,n.
\end{equation}
Furthermore, allocations $\psi_{i}^\Lambda(\cdot)$ define a generalised Euler allocation rule for lambda quantiles with the full allocation property:
\begin{equation}
\sum_{i=1}^n\psi_{i}^\Lambda(g_{\bm{X}})=\rho_{\Lambda}(\bm{1}).
\end{equation}
\end{proposition}
\begin{proof}
The lambda quantile can be written as the composition $ \rho_\Lambda = \Phi_{\Lambda}\circ g_{\bm{X}}\colon U \to \mathbb{R} \cup \{+ \infty\}$, where we define $\Phi_{\Lambda}$ for fixed $\bm{X}\in\mathcal{X}^n$ by $\Phi_{\Lambda}(g_{\bm{X}}(\bm{u}))=\rho_{\Lambda}(\bm{u})$, for all $\bm{u} \in U$.
Moreover, the lambda quantile is positively homogeneous of degree $\tau(\bm{u})\; \eta_{\Lambda,g_{\bm{X}}(\bm{u})}(-\rho_{\Lambda}(\bm{u}))$ by Proposition \ref{prop:homo-lambda}. We obtain, using Theorem \ref{thm:tlqrc} in the third equality, that
\begin{align*}
\psi_i^\Lambda(g_{\bm{X}})
&=\frac{1}{\tau(\bm{1})\,\eta_{\Lambda,g_{\bm{X}}(\bm{1})}(-\rho_{\Lambda}(\bm{1})))}\frac{\partial\Phi_{\Lambda}}{\partial u_i}(g_{\bm{X}}(\bm{u}))\biggl|_{\bm{u}=\bm{1}}
\\
&=\frac{1}{\tau(\bm{1})\,\eta_{\Lambda,g_{\bm{X}}(\bm{1})}(-\Phi_{\Lambda}(g_{\bm{X}}(\bm{1})))}\frac{\partial\rho_{\Lambda}}{\partial u_i}(\bm{u})\biggl|_{\bm{u}=\bm{1}}
\\
&=-\frac{\eta_{\Lambda,g_{\bm{X}}(\bm{1})}(-\Phi_{\Lambda}(g_{\bm{X}}(\bm{1})))}{\tau(\bm{1})\,\eta_{\Lambda,g_{\bm{X}}(\bm{1})}(-\Phi_{\Lambda}(g_{\bm{X}}(\bm{1})))}\mathbb{E}[\partial_{u_i} g_{\bm{X}}(\bm{1})~|~g_{\bm{X}}(\bm{1})=-\rho_{\Lambda}(\bm{1})]\\
&=-\frac{1}{\tau(\bm{1})}\mathbb{E}[\partial_{u_i} g_{\bm{X}}(\bm{1})~|~g_{\bm{X}}(\bm{1})=-\rho_{\Lambda}(\bm{1})].
\end{align*}
Observe that by the $\mathbb{P}$-a.s. $\tau$-homogeneity property of $g_{\bm{X}}$, we can write
\begin{align*}
\sum_{i=1}^n\psi_i^\Lambda(g_{\bm{X}})&=-\frac{1}{\tau(\bm{1})}\sum_{i=1}^n\mathbb{E}[\partial_{u_i} g_{\bm{X}}(\bm{1})~|~g_{\bm{X}}(\bm{1})=-\rho_{\Lambda}(\bm{1})]\\
&=-\frac{1}{\tau(\bm{1})}\mathbb{E}\biggl[\sum_{i=1}^n\partial_{u_i} g_{\bm{X}}(\bm{1})~|~g_{\bm{X}}(\bm{1})=-\rho_{\Lambda}(\bm{1})\biggl]\\
&=-\frac{1}{\tau(\bm{1})}\mathbb{E}[\tau(\bm{1}) g_{\bm{X}}(\bm{1})~|~g_{\bm{X}}(\bm{1})=-\rho_{\Lambda}(\bm{1})]\\
&=\rho_{\Lambda}(\bm{1})=\Phi_{\Lambda}(g_{\bm{X}}(\bm{1})).
\end{align*}
\end{proof}
In applications, portfolio operators and their portfolios are typically positively homogeneous of a constant degree, that is $\tau(\bm{u}) = \tau \in \mathbb{R}$. Thus, the multiplicative factor in the Euler contributions $\frac{1}{\tau(\bm{1})}$ reduces to $\frac{1}{\tau}$. 
\begin{example}
For Example \ref{exmp:ops} \eqref{mj:ii}, the generalised Euler contributions of the lambda quantile are given by
\begin{equation}
    \psi_{i}^\Lambda(g_{\bm{X}})
    =
    - \mathbb{E}\left[\max\left\{0, X_i - \mathbb{E}[X_i]\right\} ~|~ g_{\bm{X}}(\bm{1}) = -\rho_\Lambda(\bm{1}; g_{\bm{X}})\right]\,.
\end{equation}
For Example \ref{exmp:ops} \eqref{mj:iii}, the risk contributions of the lambda quantile are given by
\begin{align}
    \psi_{i}^\Lambda(g_{\bm{X}})
    &= -\mathbb{E}[X_i~|~X_1 + X_2=-\rho_{\Lambda}(\bm{1};g_{\bm{X}}) + VaR_{\lambda}(\bm{1};X_1 + X_2)]
    \\
    & \quad  - \mathbb{E}[X_i~|~X_1 + X_2 =-VaR_{\lambda}(\bm{1};X_1 + X_2)]\,.
\end{align}
\end{example}

Euler allocations and their desirable properties are typically considered for 1-homogeneous risk measures defined on the space of random variables and linear portfolio operators. An exception is \cite{Pesenti2021RA} who consider non-linear but positively homogeneous portfolios and distortion risk measures which are 1-homogeneous. If we consider the special case of a linear portfolio operator, then the generalised Euler allocations of the lambda quantile fulfil the properties of \textit{monotonicity} and \textit{risklessness}. Monotonicity is the property that if $X_j \ge X_i$ $\mathbb{P}$-a.s., then the generalised Euler contributions of $X_j$ is smaller than the contribution of $X_i$, i.e. $\psi^\Lambda_j(g_{\bm{X}}) \le \psi^\Lambda_i(g_{\bm{X}})$. An Euler allocation is called riskless, if $X_j$ is $\mathbb{P}$-a.s. constant, $X_j= a$, $a\in \mathbb{R}$, then $\psi^\Lambda_j(g_{\bm{X}}) = -a$. We refer to \cite{d01} and \cite{kalkbrener2005MF} for a detailed discussion of properties of Euler allocations for linear portfolios. For non-linear portfolios, however, the monotonicity and the riskless property do not hold in general nor are they desirable. Indeed consider the  portfolio $g_{\bm{X}}(\bm{u}) = u_1 X_1  + u_2 X_1 X_2$, then the generalised Euler contribution are
\begin{align}
    \psi_1^\Lambda(g_{\bm{X}}) &=  -\mathbb{E}[X_1 ~|~ g_{\bm{X}}(\bm{1}) = -\rho_\Lambda (\bm{1})\,
    \\
    \psi_2^\Lambda(g_{\bm{X}}) &=  -\mathbb{E}[X_1 X_2 ~|~ g_{\bm{X}}(\bm{1}) = -\rho_\Lambda (\bm{1})]\,.
\end{align}
Therefore, a stochastic ordering of $X_1$ and $X_2$ should not induce a ordering of the risk contributions. Moreover, if $X_2 = a$, for some $a \in \mathbb{R }$, then the generalised Euler contributions to $X_2$ is given by $\psi_2^\Lambda(g_{\bm{X}}) = a\, \psi_1^\Lambda(g_{\bm{X}})  \neq -a$.
\section{Homogeneity of portfolio operators}\label{hpo}
In applications, portfolio operators and their portfolios are typically positively homogeneous of a constant degree, that is $\tau(\bm{u}) = \tau>0$. Thus, for simplicity of exposition, we assume throughout this section that the considered portfolios are positively $\tau$-homogeneous of constant degree $\tau>0$.
The central assumption for the Euler decomposition of lambda quantiles is the $\mathbb{P}$-a.s. positively $\tau$-homogeneity of $g_{\bm{X}}$. Thus, in this section, we study properties that ensure $\mathbb{P}$-a.s. positively homogeneity in $U$ of 
generic portfolio operators. For this, we first consider operators $g$ of the following additive form to motivate some preliminary results:
\begin{equation}\label{eq:gab}
g[\bm{u},\bm{X}]=\mathfrak{a}[\bm{u},\bm{X}]+\mathfrak{b}(\bm{u},\bm{X}),
\end{equation}
where $\mathfrak{a}:U\times\mathcal{X}^n\rightarrow\mathcal{X}$ and $\mathfrak{b}:U\times\mathcal{X}^n\rightarrow\mathbb{R}$. We refer to $\mathfrak{a}$ as the \emph{stochastic} part of $g$ because it depends on a given $\omega\in\Omega$ and $\mathfrak{b}$ as the \emph{deterministic} part of $g$ because it is a constant over all choices of $\omega\in\Omega$ (in \cite{m18}, $\mathfrak{a}$ and $\mathfrak{b}$ are referred to as the pointwise and constant functions respectively). In what follows and unless otherwise stated, homogeneity of $g$ and $\mathfrak{a}$ is understood in the $\mathbb{P}$-a.s. sense (see Definition \ref{defn:esh}), whereas homogeneity of $\mathfrak{b}$ and $\rho_{\Lambda}$ is understood in the sense of Definition \ref{def:hom}.

\begin{proposition}\label{propo:operab}
Suppose the portfolio operator $g$ can be written in the form \eqref{eq:gab}. Then, $g$ is $\mathbb{P}$-a.s. positively $\tau$-homogeneous in $U$, $\tau \in \mathbb{R}$, if $\mathfrak{a}$ is $\mathbb{P}$-a.s. positively $\tau$-homogeneous in $U$ and $\mathfrak{b}$ is positively $\tau$-homogeneous in $U$.
\end{proposition}

\begin{proof}
If $\mathfrak{a}$ is $\mathbb{P}$-a.s. positively $\tau$-homogeneous in $U$, then for almost all $\omega\in\Omega$ and for any $t>0$ and $\bm{u}\in U$ with $t\bm{u}\in U$, we have:
\begin{align*}
g[t\bm{u},\bm{X}](\omega)=&\mathfrak{a}[t\bm{u},\bm{X}](\omega)+\mathfrak{b}(t\bm{u},\bm{X})\\
=&t^{\tau}\mathfrak{a}[\bm{u},\bm{X}](\omega)+t^{\tau}\mathfrak{b}(\bm{u,X})\\
=&t^{\tau}g[\bm{u},\bm{X}](\omega),
\end{align*}
and hence $g$ is $\mathbb{P}$-a.s. positively $\tau$-homogeneous.
\end{proof}

\begin{proposition}
Suppose the portfolio operator $g$ can be written in the form \eqref{eq:gab} and $\mathfrak{a}$ is given by:
\begin{equation*}
\mathfrak{a}[\bm{u},\bm{X}]=\sum_{i=1}^n u_i^{\tau}h_i(X_i),
\end{equation*}
where $h_i:\mathcal{X}\rightarrow\mathcal{X}$ for $i=1,\dots,n$ and $\tau\in\mathbb{R}$. Then, $g$ is $\mathbb{P}$-a.s. positively $\tau$-homogeneous in $U$ if, and only if, $\mathfrak{b}$ is positively $\tau$-homogeneous in $U$.
\end{proposition}

\begin{proof}
We note that $\mathfrak{a}$ is $\mathbb{P}$-a.s. positively $\tau$-homogeneous in $U$ because of the powers of the $u_i$'s. If $\mathfrak{b}$ is also positively $\tau$-homogeneous, then $g$ is $\mathbb{P}$-a.s. positively $\tau$-homogeneous by an argument similar to the proof of Proposition \ref{propo:operab}. For the opposite case, assume $g$ is $\mathbb{P}$-a.s. positively $\tau$-homogeneous. Then for any $t>0$ and $\bm{u}$ with $t\bm{u}\in U$ and almost all $\omega\in\Omega$, we have:
\begin{align*}
g[t\bm{u},\bm{X}](\omega)&=\sum_{i=1}^n (u_it)^{\tau}h_i(X_i(\omega))+\mathfrak{b}(t\bm{u},\bm{X})\\
&=t^{\tau}\mathfrak{a}[\bm{u},\bm{X}](\omega)+\mathfrak{b}(t\bm{u},\bm{X}).
\end{align*}
By our assumption, we can write:
\begin{equation*}
g[t\bm{u},\bm{X}](\omega)=t^{\tau}g[\bm{u},\bm{X}](\omega)=t^{\tau}(\mathfrak{a}[\bm{u},\bm{X}](\omega)+\mathfrak{b}(\bm{u,X})).
\end{equation*}
Hence, $\mathfrak{b}(t\bm{u,X})=t^{\tau}\mathfrak{b}(\bm{u,X})$ and $\mathfrak{b}$ is positively $\tau$-homogeneous in $U$.
\end{proof}

\begin{corollary}\label{cor:linlq}
An operator of the form $g=\mathfrak{a}[\bm{u},\bm{X}]+\rho_{\Lambda}(\bm{u})$, where $\mathfrak{a}$ is linear in $\bm{u}$ and $\Lambda$ is strictly increasing, is not homogeneous in $U$.
\end{corollary}

The assumption that $\Lambda$ is strictly increasing in Corollary \ref{cor:linlq} implies that the homogeneity degree of the lambda quantile is strictly greater than 1. Therefore, by Proposition \ref{propo:operab}, operator $g$ is not positively homogeneous.

\begin{example}
Consider the operators (\ref{op:lin})-(\ref{mj:iv}) from Example \ref{exmp:ops}. These operators are $\mathbb{P}$-a.s. 1-homogeneous in $U$ since the $\min$ and $\max$ functions and the VaR are 1-homogeneous. \\
Operator (\ref{op:lq}) from Example \ref{exmp:ops} is $\mathbb{P}$-a.s. positively homogeneous of degree $\tau \in \mathbb{R}$ in $U$ if, and only if $1=\eta_{\Lambda,Y}(-\rho_{\Lambda}(\bm{u}))$, which means that the $\Lambda^\prime(-\rho_{\Lambda}(\bm{u})) =0$ is  for all $\bm{u} \in U$.
\end{example}

The next result shows that if a deterministic variable $\mathfrak{b}_{\bm{X}}(\bm{u})$ (of portfolio compositions), i.e. a positive cash amount determined by asset units, is added to, or subtracted from, the portfolio profit and loss, then the lambda quantile is reduced or increased, respectively, by the same amount $\mathfrak{b}_{\bm{X}}(\bm{u})$. Indeed, the lambda function is shifted by this deterministic variable. The following result is related to the $\Lambda$-translation invariance property of lambda quantiles in \cite{fmp14}, that is the cash additivity property of lambda quantiles defined on the set of probability measures. However, it should not be confused with the translation invariance property defined on $U$, since we are not adding or subtracting from the portfolio composition $\bm{u}$ but from $\mathfrak{a}_{\bm{X}}(\bm{u})$ instead.

\begin{proposition}\label{prop:tinv}
For fixed $\bm{X}\in\mathcal{X}^n$, consider the portfolio: 
\begin{equation}\label{eq:gxab}
g_{\bm{X}}(\bm{u})=\mathfrak{a}_{\bm{X}}(\bm{u})+\mathfrak{b}_{\bm{X}}(\bm{u}),
\end{equation}
where $\mathfrak{a}_{\bm{X}}(u) := \mathfrak{a}[u, \bm{X}]$ and $\mathfrak{b}_{\bm{X}}(u) := \mathfrak{b}(u, \bm{X})$. Also, define $\Gamma(z):=\Lambda(z+\mathfrak{b}_{\bm{X}}(\bm{u}))$ for all $z\in\mathbb{R}$. Then, for all $\bm{u}\in U$, it holds that:
\begin{equation*}
\rho_{\Lambda}(\bm{u};g_{\bm{X}})=\rho_{\Gamma}(\bm{u};\mathfrak{a}_{\bm{X}})-\mathfrak{b}_{\bm{X}}(\bm{u}).
\end{equation*}
\end{proposition}

\begin{proof}
Observe that for all $\bm{u}\in U$, we can write:
\begin{align*}
\rho_{\Lambda}(\bm{u};g_{\bm{X}})&:=-\inf\{y~|~\mathbb{P}[\mathfrak{a}_{\bm{X}}(\bm{u})+\mathfrak{b}_{\bm{X}}(\bm{u})\leq y]>\Lambda(y)\}
\\
&=-\inf\{y~|~\mathbb{P}[\mathfrak{a}_{\bm{X}}(\bm{u})\leq y-\mathfrak{b}_{\bm{X}}(\bm{u})]>\Lambda(y)\}
\\
&=-\inf\{z+\mathfrak{b}_{\bm{X}}(\bm{u})~|~\mathbb{P}[\mathfrak{a}_{\bm{X}}(\bm{u})\leq z]>\Lambda(z+\mathfrak{b}_{\bm{X}}(\bm{u}))\}
\\
&=-\inf\{z~|~\mathbb{P}[\mathfrak{a}_{\bm{X}}(\bm{u})\leq z]>\Gamma(z)\}-\mathfrak{b}_{\bm{X}}(\bm{u})\\
&=\rho_{\Gamma}(\bm{u};\mathfrak{a}_{\bm{X}})-\mathfrak{b}_{\bm{X}}(\bm{u})\,.
\end{align*}
\end{proof}

In Proposition \ref{prop:tinv}, $\mathfrak{b}_
{\bm{X}}(\bm{u})$ corresponds to the cash amount determined by the asset units $\bm{u}$, which is added to the existing portfolio $\mathfrak{a}_{\bm{X}}(\bm{u})$ to obtain the $g_{\bm{X}}(\bm{u})$ of the newly formed portfolio. As a result, the risk of the new portfolio, i.e. $\rho_{\Lambda}(\bm{u};g_{\bm{X}})$, is obtained by subtracting $\mathfrak{b}_{\bm{X}}(\bm{u})$ from the risk of the existing portfolio, i.e. $\rho_{\Gamma}(\bm{u};\mathfrak{a}_{\bm{X}})$. Note that the cash injection causes the lambda quantile of existing and new portfolios to be calculated using different, but related lambda functions, with the relationship given by $\Gamma(z):=\Lambda(z+\mathfrak{b}_{\bm{X}}(\bm{u}))$.

\begin{proposition}\label{prop:homogeneity-non-linear}
Suppose $g_{\bm{X}}$ can be written in the form \eqref{eq:gxab}. If $\rho_{\Gamma}(\bm{u};\mathfrak{a}_{\bm{X}})$ and $\mathfrak{b}_{\bm{X}}$ are positively homogeneous in $U$ with the same degree, then $\rho_{\Lambda}(\bm{u};g_{\bm{X}})$ is positively homogeneous in $U$ with degree:
\begin{equation}
\eta_{\Lambda,Y}(-\rho_{\Lambda}(\bm{u};g_{\bm{X}}))=\eta_{\Gamma,Z}(-\rho_{\Gamma}(\bm{u};\mathfrak{a}_{\bm{X}})),
\end{equation}
where $Y=g_{\bm{X}}(\bm{u})$ and $Z:=\mathfrak{a}_{\bm{X}}(\bm{u})$.
\end{proposition}
\begin{proof}
Suppose $\rho_{\Gamma}(\bm{u};\mathfrak{a}_{\bm{X}})$ and $\mathfrak{b}_{\bm{X}}$ have homogeneity degree $\tau\in\mathbb{R}$. Then, by Proposition \ref{prop:tinv} we can write the following for any $t>0$:
\begin{align*}
\rho_{\Lambda}(t\bm{u};g_{\bm{X}})&=\rho_{\Gamma}(t\bm{u};\mathfrak{a}_{\bm{X}})-\mathfrak{b}_{\bm{X}}(t\bm{u})\\
&=t^{\tau}\rho_{\Gamma}(\bm{u};\mathfrak{a}_{\bm{X}})-t^{\tau}\mathfrak{b}_{\bm{X}}(\bm{u})\\
&=t^{\tau}\rho_{\Lambda}(\bm{u};g_{\bm{X}}),
\end{align*}
which implies $\rho_{\Lambda}(\bm{u};g_{\bm{X}})$ is positively $\tau$-homogeneous in $U$. Thus by  Equation \eqref{eq: Euler}, both $\rho_{\Lambda}(\bm{u};g_{\bm{X}})$ and $\rho_{\Gamma}(\bm{u};\mathfrak{a}_{\bm{X}})$ can be written in the form (\ref{lqed}) from Theorem \ref{thm:ed}. Therefore, the homogeneity degrees of $\rho_{\Lambda}(\bm{u};g_{\bm{X}})$ and $\rho_{\Gamma}(\bm{u};\mathfrak{a}_{\bm{X}})$ are given by $\eta_{\Lambda,Y}(-\rho_{\Lambda}(\bm{u};g_{\bm{X}}))$ and $\eta_{\Gamma,Z}(-\rho_{\Gamma}(\bm{u};\mathfrak{a}_{\bm{X}}))$ respectively, which proves our result.
\end{proof}
Proposition \ref{prop:homogeneity-non-linear} in particular applies to Example \ref{exmp:ops} cases \eqref{mj:i}, \eqref{mj:iii}, and \eqref{op:lq}, with $\mathfrak{a}_{\bm{X}}(\bm{u}) = u_1X_1 +u_2 X_2 $ for \eqref{mj:i} and \eqref{mj:iii}, and with $\mathfrak{a}_{\bm{X}}(\bm{u}) = u_1^\tau X_1 +u_2^\tau X_2 $ for example \eqref{op:lq}. The components $\mathfrak{b}_{\bm{X}}$ are given by $\mathfrak{b}_{\bm{X}}(\bm{u}) = -\mathbb{E}[ u_1X_1 +u_2 X_2 ]$, $\mathfrak{b}_{\bm{X}}(\bm{u}) = -VaR_\lambda(\bm{u};\mathfrak{a}_{\bm{X}})$, and $\mathfrak{b}_{\bm{X}}(\bm{u}) = - \rho_\Lambda(\bm{u};\mathfrak{a}_{\bm{X}})$, respectively.

\begin{proposition}
Let $g_{\bm{X}}$ be of the form \eqref{eq:gxab} and assume:
\begin{enumerate}[(i)]
\item Assumption \ref{asmp:t} is satisfied for $\mathfrak{a}_{\bm{X}}$ and $\Gamma$;
\item $\Gamma$ is continuously differentiable on $\mathbb{R}$;
\item $\mathfrak{b}_{\bm{X}}$ is continuously differentiable by $u_i$ for $i=1,\ldots,n$ for all $\bm{u}\in U$.
\end{enumerate}
Then, $\rho_{\Lambda}(\cdot;g_{\bm{X}})$ is partially differentiable in $U$ with continuous derivatives given by:
\begin{equation}
\frac{\partial \rho_{\Lambda}}{\partial u_i}(\bm{u};g_{\bm{X}})=-\eta_{\Gamma,Z}(-\rho_{\Gamma}(\bm{u};\mathfrak{a}_{\bm{X}}))\mathbb{E}[\partial_{u_i}\mathfrak{a}_{\bm{X}}(\bm{u})|Z=-\rho_{\Gamma}(\bm{u};\mathfrak{a}_{\bm{X}})]-\partial_{u_i}\mathfrak{b}(\bm{u}),
\end{equation}
for $i=1,\dots,n$.
\end{proposition}
\begin{proof}
By Proposition \ref{prop:tinv}, we can write
\begin{equation*}
\rho_{\Lambda}(\bm{u};g_{\bm{X}})=\rho_{\Gamma}(\bm{u};\mathfrak{a}_{\bm{X}})-\mathfrak{b}_{\bm{X}}(\bm{u}).
\end{equation*}
By Theorem \ref{thm:tlqrc} with condition \ref{thmasmp:t}, $\rho_{\Gamma}(\cdot;\mathfrak{a}_{\bm{X}})$ is partially differentiable in $U$ with continuous derivatives given by:
\begin{equation*}
\frac{\partial \rho_{\Gamma}}{\partial u_i}(\bm{u};\mathfrak{a}_{\bm{X}})=-\eta_{\Gamma,Z}(-\rho_{\Gamma}(\bm{u};\mathfrak{a}_{\bm{X}}))\mathbb{E}[\partial_{u_i} \mathfrak{a}_{\bm{X}}(\bm{u})|Z=-\rho_{\Gamma}(\bm{u};\mathfrak{a}_{\bm{X}})],
\end{equation*}
for $i=1,\dots,n$. Since both $\rho_{\Gamma}(\cdot;\mathfrak{a}_{\bm{X}})$ and $\mathfrak{b}_{\bm{X}}$ are continuously differentiable in $U$, we conclude that $\rho_{\Lambda}(\cdot;g_{\bm{X}})$ is continuously partially differentiable in $U$.
\end{proof}

As seen in Examples \ref{exmp:ops}, the function $\mathfrak{a}$ is typically a polynomial in $\bm{X}$ whilst $\mathfrak{b}$ an expectation or lambda quantile. A portfolio operator may also be constructed via a function of these ``building blocks". Thus, we consider $g$ to be a composition of a function $f:\mathcal{X}\times\mathbb{R}\rightarrow\mathcal{X}$ with $\mathfrak{a}$ and $\mathfrak{b}$, i.e. we consider operators of the form
\begin{equation}\label{eqn:opr}
g[\bm{u},\bm{X}]=(f\circ(\mathfrak{a},\mathfrak{b}))(\bm{u,X})=f(\mathfrak{a}[\bm{u},\bm{X}],\mathfrak{b}(\bm{u,X})).
\end{equation}
Since $f$ acts on $\mathfrak{a}[\bm{u},\bm{X}]\in\mathcal{X}$ and $\mathfrak{b}(\bm{u,X})\in\mathbb{R}$, positively homogeneity of the function $f$ is discussed in $\mathcal{X}$ and $\mathbb{R}$, but not in $U$. The function $f$ in this case may be implicitly  positively homogeneous in $U$.

\begin{proposition}
Suppose the function $\mathfrak{a}:U\times\mathcal{X}^n\rightarrow\mathcal{X}$ is $\mathbb{P}$-a.s. positively $\tau$-homogeneous in $U$ and $\mathfrak{b}:U\times\mathcal{X}^n\rightarrow\mathbb{R}$ is positively $\tau$-homogeneous in $U$. Also, suppose the function $f:\mathcal{X}\times\mathbb{R}\rightarrow\mathcal{X}$ is $\mathbb{P}$-a.s. positively $\nu$-homogeneous in both $\mathcal{X}$ and $\mathbb{R}$. Then, the operator in Equation \eqref{eqn:opr} is $\mathbb{P}$-a.s. positively homogeneous of degree $\tau\nu$ in $U$.
\end{proposition}

\begin{proof}
Noting that $\mathfrak{b}(\bm{u,X})$ is constant across outcomes $\omega\in\Omega$, we can write the following for almost all $\omega\in\Omega$ and for any $t>0$ and $\bm{u}\in U$:
\begin{align*}
g[t\bm{u,X}](\omega)&=(f\circ(\mathfrak{a},\mathfrak{b}))(t\bm{u,X})(\omega)\\
&=f(\mathfrak{a}[t\bm{u,X}](\omega),\mathfrak{b}(t\bm{u,X}))\\
&=f(t^{\tau}\mathfrak{a}[\bm{u},\bm{X}](\omega),t^{\tau}\mathfrak{b}(\bm{u,X}))\\
&=(t^{\tau})^{\nu}f(\mathfrak{a}[\bm{u},\bm{X}](\omega),\mathfrak{b}(\bm{u,X}))\\
&=t^{\tau\nu}f(\mathfrak{a}[\bm{u},\bm{X}](\omega),\mathfrak{b}(\bm{u,X}))\\
&=t^{\tau\nu}(f\circ(\mathfrak{a},\mathfrak{b}))(\bm{u,X})(\omega)\\
&=t^{\tau\nu}g[\bm{u},\bm{X}](\omega).
\end{align*}
Hence, $g$ is $\mathbb{P}$-a.s. positively homogeneous in $U$ of degree $\tau\nu$.
\end{proof}
\vspace{-2em}
\section{Applications to financial markets}\label{sect:applications}

In this section we present an application using real market data to illustrate the lambda quantile's homogeneity degree, its generalised Euler and risk contributions on a non-linear but homogeneous portfolio. \\
\subsection{Portfolio setup}
We consider a non-linear portfolio where we measure the portfolio risk using lambda quantiles; a study which has so far not been conducted in the literature. For this purpose, we construct a portfolio consisting of an stock and a European call option whose underlying asset is the same stock. 
Our portfolio is non-linear since the profit and loss of an option is a non-linear function of the underlying asset's profit and loss.

We denote the prices of the stock and call option at the valuation date $t$ by $S_t$ and $C_t$ respectively. Furthermore, we let $\bm{u}_P=(u_s,u_c)$, where $u_s$ and $u_c$ represent the number of units of stocks and call options in our portfolio respectively. The price of the European call option at time $t \in[t_0,T]$ is calculated using the Black-Scholes formula, where $t_0$ is the option issue date and $T$ the option maturity.

For this numerical application, we consider the stock of Exxon Mobil Corporation (NYSE: XOM) and an option on the same stock with a maturity of 2 years ($T =  2 \times 250$ days), and strike at moneyness level $90\%$ of the stock price. The daily stock market close prices for the period 1\textsuperscript{st} January 2018 to 5\textsuperscript{th} November 2021 have been sourced from Bloomberg. The option issue date $t_0$ is 13\textsuperscript{th} November 2018. In our implementation, we compute the daily lambda quantiles of our portfolio for all days in the period $[t_0,T]$. At each valuation date $t \in [t_0,T]$, the time-to-maturity of the option is $T-t$ days. We assume that the number of asset units remains constant over this period, unless otherwise stated. In the Black-Scholes formula, we fix the annualised risk-free rate to be $2\%$ and use the annualised standard deviation of historical Exxon Mobil Corporation stock log-returns as the volatility.
The value of our portfolio at time $t$, $V_t(\bm{u}_P):U\rightarrow\mathcal{X}$, is thus given by:
\begin{equation*}
V_t(\bm{u}_P)=u_sS_t+u_cC_t,
\end{equation*}
and the profit and loss of the portfolio, $PL_t(\bm{u}_P):U\rightarrow\mathcal{X}$, from time $t-1$ to time $t$ is defined by
\begin{equation}\label{eq:pnl}
PL_t(\bm{u}_P):=V_t(\bm{u}_P)-V_{t-1}(\bm{u}_P)\,.
\end{equation}
Next, we rewrite the profit and loss in the same form as Example \ref{exmp:ops} (\ref{op:lin}). For this, we define the profit and loss of the stock and call option at time $t$ by $X_{s,t}:=S_t-S_{t-1}$ and $X_{c,t}:=C_t-C_{t-1}$, respectively, such that $\bm{X}_t:=(X_{s},X_{c})_t$, and 
\begin{equation*}
PL_t(\bm{u}_P)=g_{\bm{X}_t}(\bm{u}_P)=u_sX_{s,t}+u_cX_{c,t}.
\end{equation*}
Note that the $PL_t$ is a homogeneous function in $U$ of degree 1, thus the lambda quantile applied to the portfolio is homogeneous with degree given in Proposition \ref{prop:homo-lambda}. Partial derivatives of $PL_t$ with respect to asset units, which appear in the partial derivatives of $\rho_{\Lambda}$ in \eqref{eqn:lqrc}, are given by:
\begin{equation*}
\partial_{u_s}PL_t(\bm{u}_P)=X_{s,t}, \quad
\partial_{u_c}PL_t(\bm{u}_P)=X_{c,t}.
\end{equation*}
To estimate the probability density and the distribution function (CDF) of the portfolio profit and loss,
we use the full valuation method based on historical simulation of the underlying asset's profits and losses combined with kernel density estimation (KDE). We refer to \cite{rs07} for details on how to implement the full valuation method. Specifically, we consider 250 daily historically simulated profits and losses of the stock corresponding to an in-sample period of 250 days. Then, we estimate the daily portfolio densities and CDFs via KDE with a Gaussian kernel fitted to the 250 historically simulated portfolio profits and losses.

\subsection{Estimation of the lambda function}
\begin{figure}[t]
\centering
\begin{minipage}{0.45\textwidth}
\centering
100\% stock portfolio\\
\includegraphics[width=\textwidth]{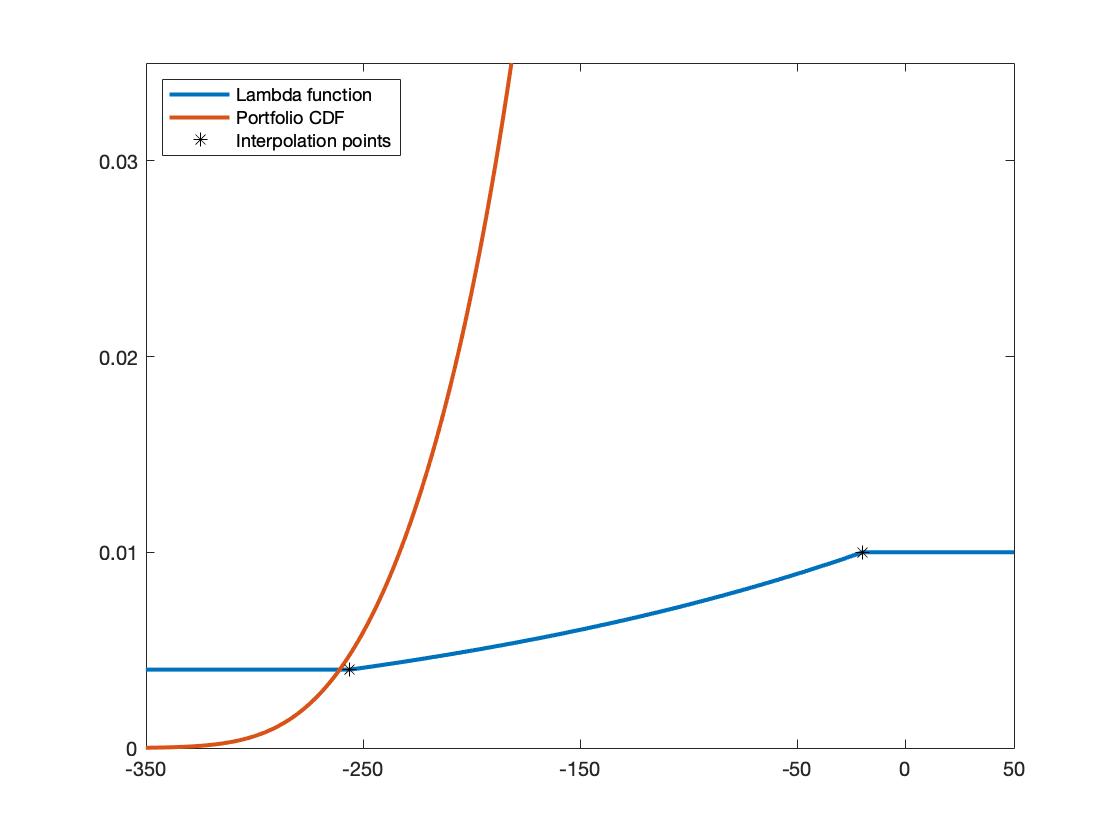}
\end{minipage}
\hfill
\begin{minipage}{0.45\textwidth}
\centering
Equally weighted portfolio\\
\includegraphics[width=\textwidth]{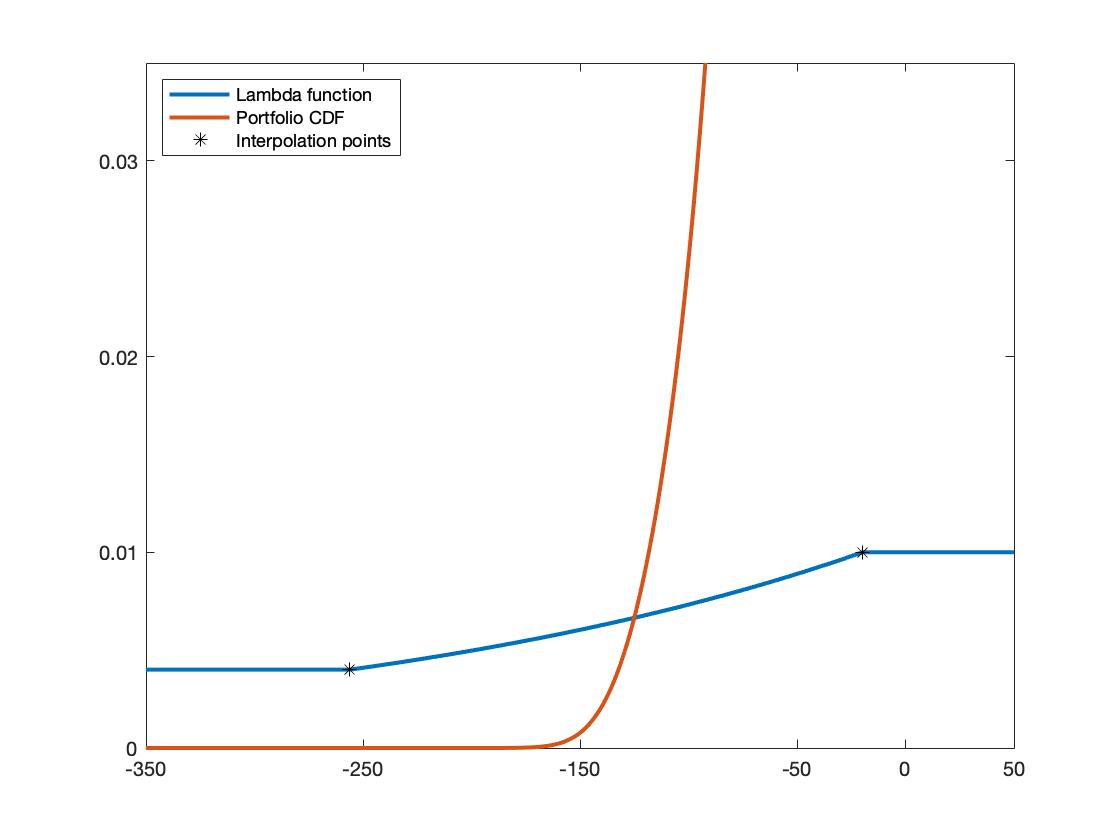}
\end{minipage}
\caption{Lambda function (blue) and portfolio CDF (red), where the $x$-axis value of the intersection point is the negative of the lambda quantile, i.e., $-\rho_{\Lambda_t}(\bm{u}_P)$. Both panels are at time 16\textsuperscript{th} July 2020. The left panel represents the portfolio with composition $u_s=100$ and $u_c=0$ and the right panel with composition $u_s=50$ and $u_c=50$.}
\label{fig:XOM_LFCDF}
\end{figure}
A first approach in the literature to estimate the lambda function is given by \cite{hmp18} which has been generalised by \cite{cp18}. In these works, the authors suggest to calibrate the lambda function via a data-driven approach based on market benchmarks, such as stock market indexes in the case of stock portfolios. 
The authors choose the maximum of the lambda function equal to a maximum acceptable level $\lambda\in(0,1)$, so that in this case the lambda quantile becomes comparable to $VaR_\lambda$.
They further set the minimum of the lambda function to be strictly larger than the frequency of one observation, see \cite{cp18} for a detailed discussion on this requirement.
The authors then obtain the lambda function as a linear interpolation between points that are associated to left-tail order statistics of the chosen market benchmarks. This approach thus treats the lambda function as a proxy of the tail distribution of market benchmarks.
\\
In this paper, we choose a parametric approach for the calibration of the lambda function that maintains the original idea in the previous literature. Specifically, we define the lambda function at time $t$, $\Lambda_t$, to have an exponential growth in the interval $[x_{a,t}, x_{b,t}]$ and to be constant otherwise. The choice of the lambda function to be convex in $[x_{a,t}, x_{b,t}]$ is prudent, since compared to a concave interpolation, a convex interpolation always results in a larger lambda quantile for any fixed portfolio. Thus, the parametric form of the lambda function is given by:
\begin{equation}\label{appl:lf}
\Lambda_t(x)= 1_{\{ x < x_{a,t}\}} \lambda_a+1_{\{ x_{a,t} \leq x \leq x_{b,t} \}} \beta_t e^{\alpha_t x}+ 1_{\{x > x_{b,t}\}} \lambda_b,
\quad x \in \mathbb{R},
\end{equation}
where $x_{a,t} < x_{b,t}$, $0<\lambda_a < \lambda_b<1$, and the coefficients $\alpha_t,\beta_t\in\mathbb{R}$ are chosen such that the lambda function is continuous, i.e., they are given by:
\begin{align*}
\alpha_t&=\frac{\log(\lambda_a/\lambda_b)}{x_{a,t}-x_{b,t}},\\
\beta_t&=\lambda_a\exp{\biggl(-\frac{x_{a,t}}{x_{a,t}-x_{b,t}}\log{(\lambda_a/\lambda_b)}\biggl)}.
\end{align*}

Note that we choose the minimum and maximum value of the lambda function, i.e., $\lambda_a$ and $\lambda_b$, respectively, to be constant over time.  Point $x_{a,t}$ represents the threshold below which a risk manager assigns to the portfolio the $VaR$ at the lowest confidence level $\lambda_a := \min_x \{\Lambda_t(x)\} = \Lambda_t(x_{a,t})$. Similar to the previous literature, $x_{a,t}$ is chosen as the daily minimum over the market index and the stock profits and losses for the given sample window. This threshold is a technical requirement so that lambda quantiles are not always lower than $x_{a, t}$, for all $t \in[t_0, T]$. We fix $\lambda_a=  1.001\times\frac{1}{250}  = 0.4004\%$ as we require that $\lambda_a $ is larger than the probability of one historically simulated observation (in-sample period of 250 days). Point $x_{b,t} := \text{arg}\max \{\Lambda_t(x)\}$ corresponds to the threshold above which a risk manager assigns to the portfolio the $VaR$ at the highest acceptable confidence level $\lambda_b := \max_x\{\Lambda_t(x)\} = \Lambda_t(x_{b,t})$, here set to $1\%$. Therefore $x_{b,t}$ is chosen, similar to the previous literature, as an extreme left-tail statistic of a market index, here the $VaR_{1\%}$ of the selected index. Note that both $\lambda_a$ and $\lambda_b$ are kept constant for all time points $t\in[t_0,T]$.

As our portfolio consists of stock positions, we chose the S\&P500 index as the market index, where daily market close prices have been sourced from Bloomberg for the period 1\textsuperscript{st} January 2018 to 5\textsuperscript{th} November 2021. The lambda function is calibrated on a daily basis, that is $x_{a,t}$ and $x_{b,t}$ are recalculated daily using the previous 250 days. In Figure \ref{fig:XOM_LFCDF}, we display estimated lambda functions for two different dates.
Finally, at any time $t$ we compute the lambda quantile of our portfolio, here denoted $\rho_{\Lambda_t}(\bm{u}_P)$, as the negative of the smallest value on the profit and loss axis corresponding to the intersection point between the portfolio CDF and the lambda function using Matlab's \texttt{fzero} solver.

In total, we have 499 daily estimates of lambda quantiles which correspond to the 499 business days from option issuance date to $T-1 = 2 \times 250 -1$. Figure \ref{fig:XOM_LQ} shows the daily out-of-sample portfolio profits and losses and a comparison between $-\rho_{\Lambda_t}(\bm{u}_P)$ and $-VaR_{1\%,t}(\bm{u}_P)$ of the portfolio in consideration over time. As expected the lambda quantile is more conservative than $VaR_{1\%}$.

\begin{figure}[t]
\centering
\includegraphics[width=0.6\textwidth]{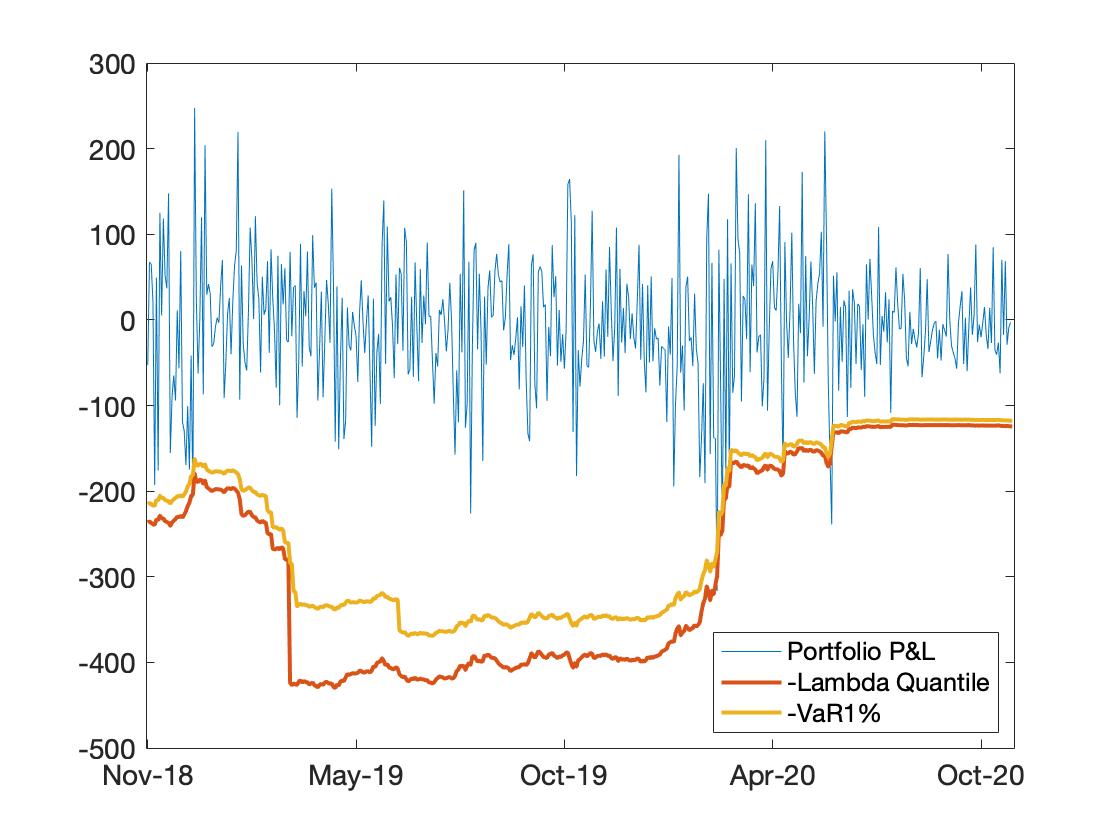}
\caption{Time series of portfolio profits and losses, negative of lambda quantile and negative of $VaR_{1\%}$ of portfolio consisting of Exxon Mobil Corporation stock and call option with $m=90\%$, $T=2$-years, $u_s=u_c=50$ and issue date of 13\textsuperscript{th} November 2018.}
\label{fig:XOM_LQ}
\end{figure}

\vspace{2em}
\subsection{Homogeneity degree}\label{appl:hd}
Following the calculation of the daily lambda quantiles of our portfolio, we proceed by computing the daily homogeneity degree and show its variability over both time and the choice of portfolio composition (note that once a portfolio composition is chosen, we keep the composition fixed from inception until option maturity). Even though we keep our portfolio composition $\bm{u}_P$ fixed, we expect our portfolio's homogeneity degree to show variation over the life of the option contract since both the portfolio CDF and lambda function are re-calibrated on a daily basis. Using Proposition \ref{prop:homo-lambda}, recall that the homogeneity degree of the portfolio is equal to $\tau=1$, we can write the lambda quantile  homogeneity degree at $t$ as:

\begin{equation*}
\eta_t(\bm{u}_P)=\frac{f_t(-\rho_{\Lambda_t}(\bm{u}_P))}{f_t(-\rho_{\Lambda_t}(\bm{u}_P))-\Lambda_t'(-\rho_{\Lambda_t}(\bm{u}_P))},
\end{equation*}
where $f_t$ is the probability density function of the portfolio profit and loss $PL_t(\bm{u}_P)$ estimated on a daily basis using KDE. From this formula, it is evident that the homogeneity degree of lambda quantiles is dictated by the portfolio density $f_t$ and the slope of the lambda function $\Lambda$, both evaluated at the point $y=-\rho_{\Lambda_t}(\bm{u}_P)$. Since the lambda function is calibrated on a market index, the homogeneity degree may change day-over-day. Furthermore, the homogeneity degree may also change if a different portfolio composition is selected. Therefore, the portfolio risk, measured using the lambda quantile, and portfolio composition may not scale linearly under all market conditions, if at all.

In this numerical study, the homogeneity degree $\eta_t$ is well-defined with $\eta_t\geq1$ since by construction of the lambda function it holds that $f(-\rho_{\Lambda_t}(\bm{u}_P))>\Lambda'(-\rho_{\Lambda_t}(\bm{u}_P))$. This inequality is equivalent to Assumption \ref{asmp:h} (\ref{asmp:grad}) for our choice of portfolio profit and loss $PL_t$. The derivative of the lambda function is easily computed from \eqref{appl:lf} and given by:
\begin{equation*}
\Lambda_t'(x)=1_{\{ x_{a,t} \leq x \leq x_{b,t} \}}\alpha_t \beta_t e^{\alpha_t x},
\quad x \in \mathbb{R}.
\end{equation*}

\begin{figure}[t]
\centering
\includegraphics[width=0.6\textwidth]{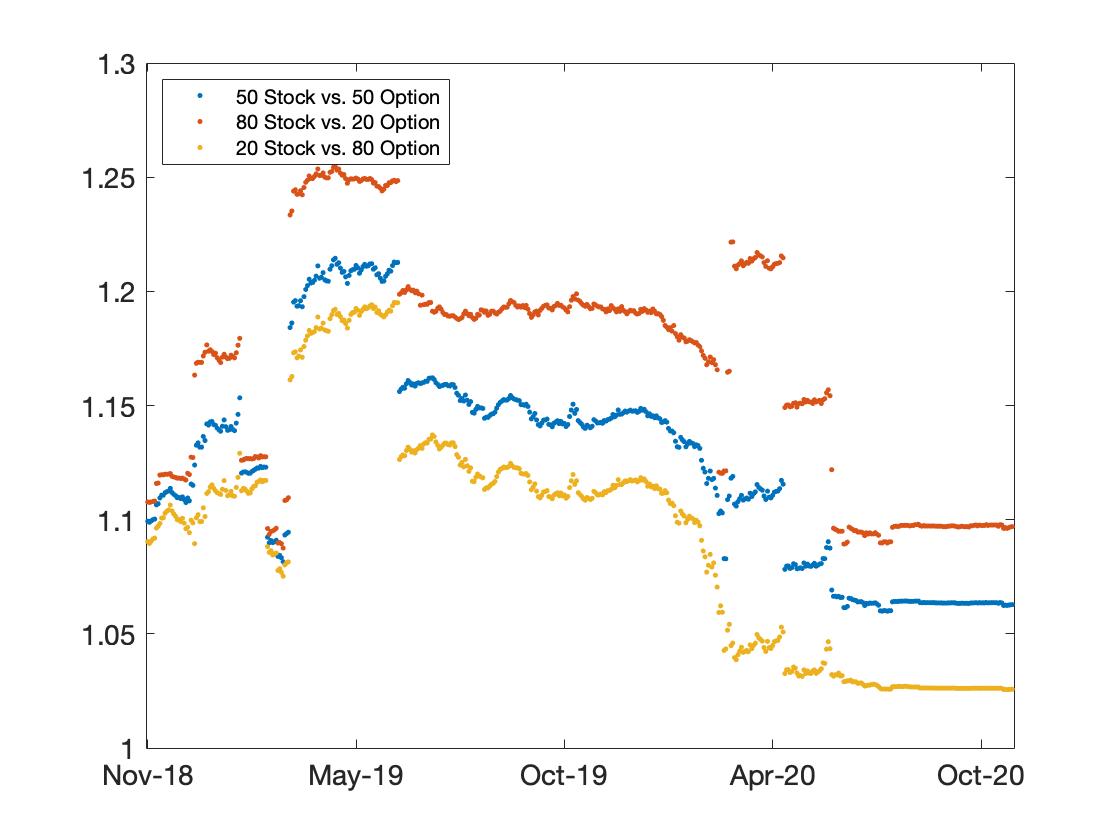}
\caption{Time series of the homogeneity degree of the lambda quantile risk measure applied on several portfolio compositions of Exxon Mobil Corporation stock and call option with $m=90\%$, $T=2$-years, and issue date of 13\textsuperscript{th} November 2018.}
\label{fig:XOM_HD}
\end{figure}

Figure \ref{fig:XOM_HD} displays the homogeneity degree of three different portfolio compositions over time. From Figure \ref{fig:XOM_HD}, we see that, in the period between 13\textsuperscript{th} November 2018 and 6\textsuperscript{th} November 2020, the homogeneity degree of all three portfolio compositions is strictly larger than 1, which occurs when $-\rho_{\Lambda_t}(\bm{u}_P)\in[x_{a,t},x_{b,t}]$. In this case, the lambda quantile of our portfolio is larger than the $VaR_{1\%}$ of the market index. The variability of the homogeneity degree of each portfolio composition is particularly visible during periods of high volatility of the Exxon Mobil Corporation stock profits and losses, as at the end of 2018 and during the initial stages of the Covid-19 pandemic. In contrast, the fluctuations of the homogeneity degree are more moderate in 2019, a period of low volatility of the Exxon Mobil Corporation stock profits and losses.

If $-\rho_{\Lambda_t}(\bm{u}_P)\in[x_{a,t},x_{b,t}]$, then the portfolio risk, measured using the lambda quantile, and portfolio composition scale with a degree larger than one, i.e. it holds that
\begin{equation}
\rho_{\Lambda_t}(k\bm{u}_P)=k^{\eta_t}\rho_{\Lambda_t}(\bm{u}_P) \quad \text{for all } k>0.
\end{equation}
Therefore, for portfolio compositions for which the lambda quantile $-\rho_{\Lambda_t}(\bm{u}_P)\in[x_{a,t},x_{b,t}]$ is closer to $x_{b,t}$, the higher is its homogeneity degree.

How the homogeneity degree is affected by different portfolio compositions is displayed in Figure \ref{fig:XOM_HD}. The equally-weighted portfolio is displayed in blue. The portfolio with 20 stocks and 80 options (displayed in yellow), has lower homogeneity degrees for the 2-years under consideration. This is in contrast to the portfolio with 80 stocks and 20 options (displayed in red) which has a larger homogeneity degree.

The homogeneity degree of the lambda quantile is equal to 1, if the portfolio CDF intersects the lambda function in its flat parts, i.e. when $-\rho_{\Lambda_t}(\bm{u}_n)\not\in[x_{a,t},x_{b,t}]$. In these cases, the lambda quantile is equal to the $VaR$ at the lowest level $\lambda_a$, if $-\rho_{\Lambda_t}(\bm{u}_n)<x_a$, or at the highest level $\lambda_b$, if $-\rho_{\Lambda_t}(\bm{u}_n)>x_b$, and the portfolio risk scales linearly with portfolio composition since:
\begin{equation}
\rho_{\Lambda_t}(k\bm{u}_P)= k\rho_{\Lambda_t}(\bm{u}_P) \quad \text{for all } k>0.
\end{equation}

We observe homogeneity degree 1 during the Covid-19 pandemic for the composition of 100 stocks and 0 calls. The reader is referred to the left plot of Figure \ref{fig:XOM_LFCDF} for an example. Here, the negative of the lambda quantile $-\rho_{\Lambda_t}(\bm{u}_n)$ is lower than $x_a$, the minimum over the market index and the stock profits and losses of the previous 250 days. This means that, the risk of the 100 stocks portfolio for this day is higher than the worst case scenario over the previous 250 days.

\subsection{Risk contributions and generalised Euler contributions}
In this section, we calculate the daily generalised Euler contributions of lambda quantile applied to the portfolio with 50 stocks and 50 options. The generalised Euler contributions quantify how much each asset in the portfolio contributes to the portfolio risk, measured by the lambda quantile risk measure. Recall that for this portfolio we have a homogeneity degree of $\tau=1$. Therefore, the generalised Euler contributions of the lambda quantile at time $t$ are given, for $j \in \{s, c\}$, by:
\begin{align}\label{condexp}
\psi_{j}^{\Lambda_t}(PL_t(\bm{1})) 
=& - \mathbb{E}[\partial_{u_j}PL_t(\bm{1})~|~PL_t(\bm{1})=-\rho_{\Lambda_t}(\bm{1})]=
- \mathbb{E}[X_{j,t}~|~PL_t(\bm{1})=-\rho_{\Lambda_t}(\bm{1})]\,.
\end{align}
We estimate the conditional expectations using the Nadaraya-Watson kernel estimator, see e.g., \cite{Wand1994book}, with a standard Gaussian kernel smoothing function. Specifically, for each time $t$, we take a sample of the historically simulated stock's profits and losses using the past $N=250$ trading days. 
At time $t$, the samples, each of size $N$, consist of the stock's, option's and portfolio's profits and losses where the $i$-th sample entry is from time $t-i$ and are given by:
\begin{align*}
x_{s,t}^{i},\quad 
x_{c,t}^{i},\quad \text{and} \quad
y_t^{i}:=u_sx_{s,t}^{i}+u_cx_{c,t}^{i},
\end{align*}
where $i=1,\dots,N$. Therefore, we estimate the generalised Euler contributions of the lambda quantile, for $j \in \{s, c\}$, by:
\begin{equation}\label{appl:nwgec}
\widehat{\psi}_{j}^{\Lambda_t}(PL_t(\bm{1}))
=- \frac{\sum_{i=1}^N x_{j, t}^i\;\phi(\frac{-\rho_{\Lambda_t}(\bm{1})-y_{t}^i}{h_t})}{\sum_{i=1}^N \phi(\frac{-\rho_{\Lambda_t}(\bm{1})-y_{t}^i}{h_t})},
\end{equation}
where $\phi$ is the standard normal density and $h_t$ is the bandwidth given by Silverman's rule, see e.g., \cite{nonpareco}:
\begin{equation*}
h_t = 1.06\times\sigma(y_t^1,\dots,y_t^N)\times N^{-1/5},
\end{equation*}
where $\sigma(y_t^1,\dots,y_t^N)$ is the standard deviation of portfolio profits and losses $y_t^1,\dots,y_t^N$ in each sample.
\begin{figure}[htb]
\centering
\begin{minipage}{0.45\textwidth}
\centering
Generalised Euler contributions\\
\includegraphics[width=\textwidth]{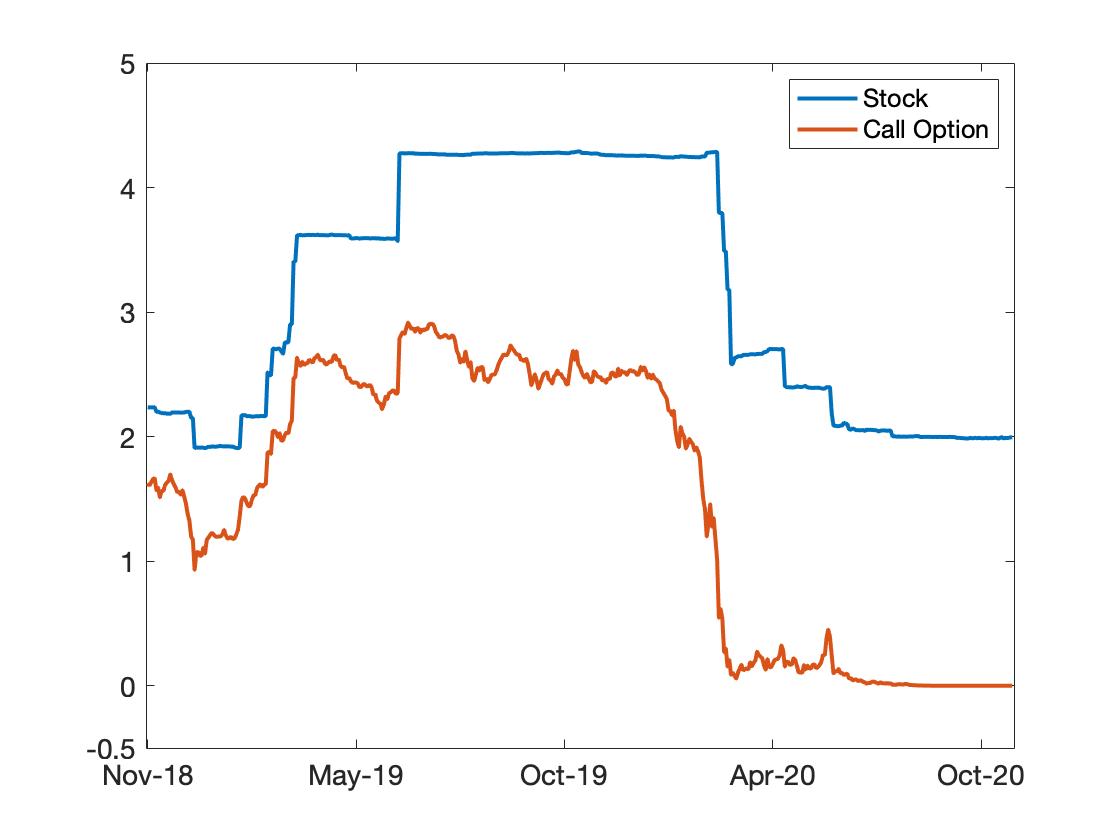}
\end{minipage}
\hfill
\begin{minipage}{0.45\textwidth}
\centering
Mean simulated P\&L\\
\includegraphics[width=\textwidth]{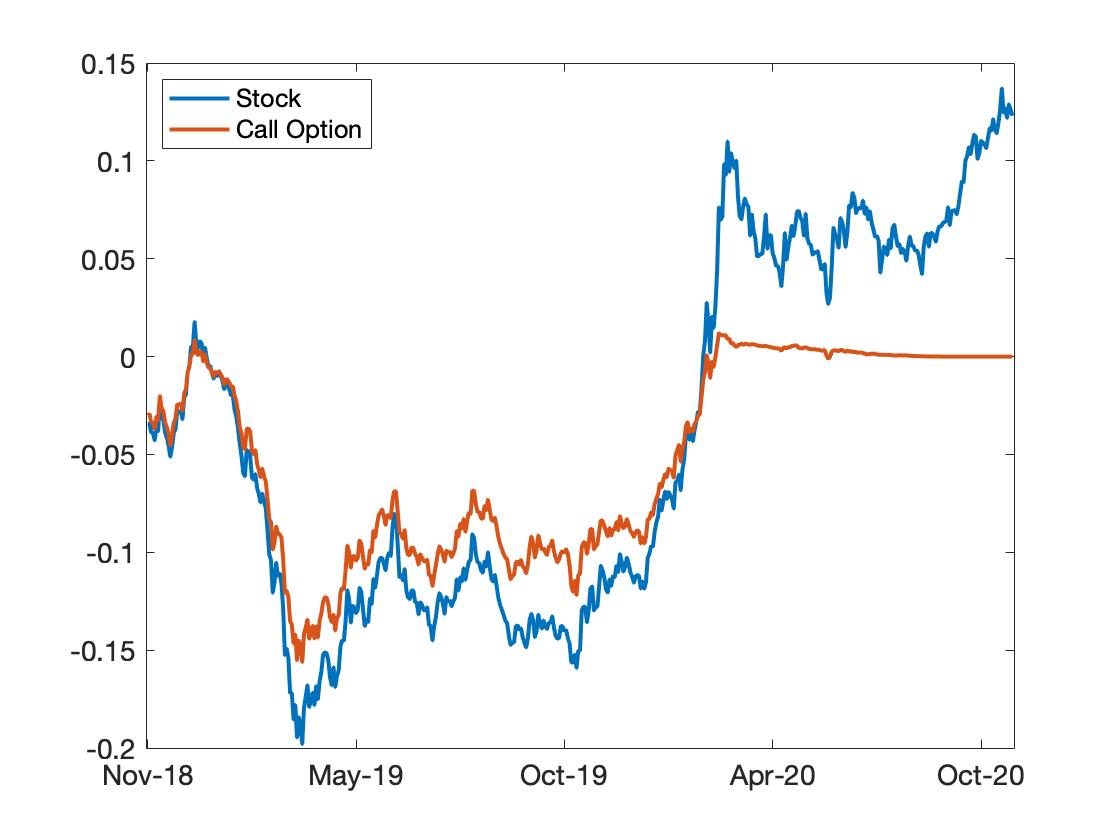}
\end{minipage}
\caption{The left panel is a comparison between the lambda quantile generalised Euler contributions of Exxon Mobil Corporation stock and call option. The right panel displays a time series of mean of simulated profits and losses of Exxon Mobil Corporation stock and call option. For both panels, we have $m=90\%$, $T=2$-years, issue date of 13\textsuperscript{th} November 2018 and $u_s=u_c=1$.}
\label{fig:XOM_GEC}
\end{figure}

The left panel of Figure \ref{fig:XOM_GEC} displays the time series of the generalised Euler contributions for the stock and call option. At time $t$, the risk contribution of the $j$-th asset represents how much the lambda quantile of the portfolio increases for an infinitesimal increase in the number of units of the $j$-th asset. We observe, that the generalised Euler contributions of the stock ranges from 1.9064 to 4.2897, while that of the call option ranges from 0 to 2.9161.\vspace*{0.3em}

Note that the estimator of the generalised Euler contributions for the stock and the call option only differ by the simulated profits and losses $x_{j,t}^i$, $j\in\{s,c\}$ since the terms $\phi(\frac{-\rho_{\Lambda_t}(\bm{1})-y_{t}^i}{h_t})$ are the same for the stock and call option. Therefore, the generalised Euler contributions of the stock are larger than the call option because the stock's profit and loss are more negative than that of the call option until the option price approaches zero, see right panel of Figure \ref{fig:XOM_GEC} where we show the time series of the average simulated profits and losses of the stock and the call option. From a financial point of view, this can be interpreted as follows: if the performance of asset 1 is poorer than that of asset 2, then increasing the exposure to asset 1 increases the lambda quantile more compared to increasing the exposure to asset 2.
 
Next, we consider the effect of varying the number of units of stocks and options in the portfolio on the lambda quantile risk contributions. To do this, we fix a valuation date and calculate risk contributions of the stock and option by considering multiple portfolio compositions. Note that for each risk contribution, the portfolio composition is kept constant from option issuance until maturity. Figure \ref{fig:XOM_RC} shows the lambda quantile risk contribution as a function of the portfolio composition for three different dates. Specifically, we plot: 
\begin{equation}
\frac{\partial \rho_{\Lambda_t} }{\partial u_j} (\bm{u}_P)
=- \eta_t(\bm{u}_P)\mathbb{E}[X_{j,t}~|~PL_t(\bm{u}_P)=-\rho_{\Lambda_t}(\bm{u}_P)]\,,
\end{equation} 
for $\bm{u}_P=(u_s,u_c)=(10,90),(11,89),\dots,(90,10)$, and where the conditional expectations is estimated using the Nadaraya-Watson kernel estimator, and the homogeneity degree $\eta_t(\bm{u}_P)$ is estimated as described in Section \ref{appl:hd}.

We observe in Figure \ref{fig:XOM_RC} that as the number of stocks in our portfolio increases, the risk contribution of the stocks also increases. In contrast, as the number of call options in our portfolio increases, the risk contribution of the call options decreases. A different pattern is observed close to the option's maturity, right panel of Figure \ref{fig:XOM_RC}. In this case the risk contributions of the call options become zero independently of the portfolio composition. This is due to the fact that close to option maturity, the option is out-of-the money and its price is zero. We further observe, from the left to the right panel in Figure \ref{fig:XOM_RC}, that the values of the risk contribution of both assets reduces as we approach option maturity.

\begin{figure}[!htb] 
\centering
\begin{minipage}{0.32\textwidth}
\centering
11\textsuperscript{th} November 2019\\
\includegraphics[width=\textwidth]{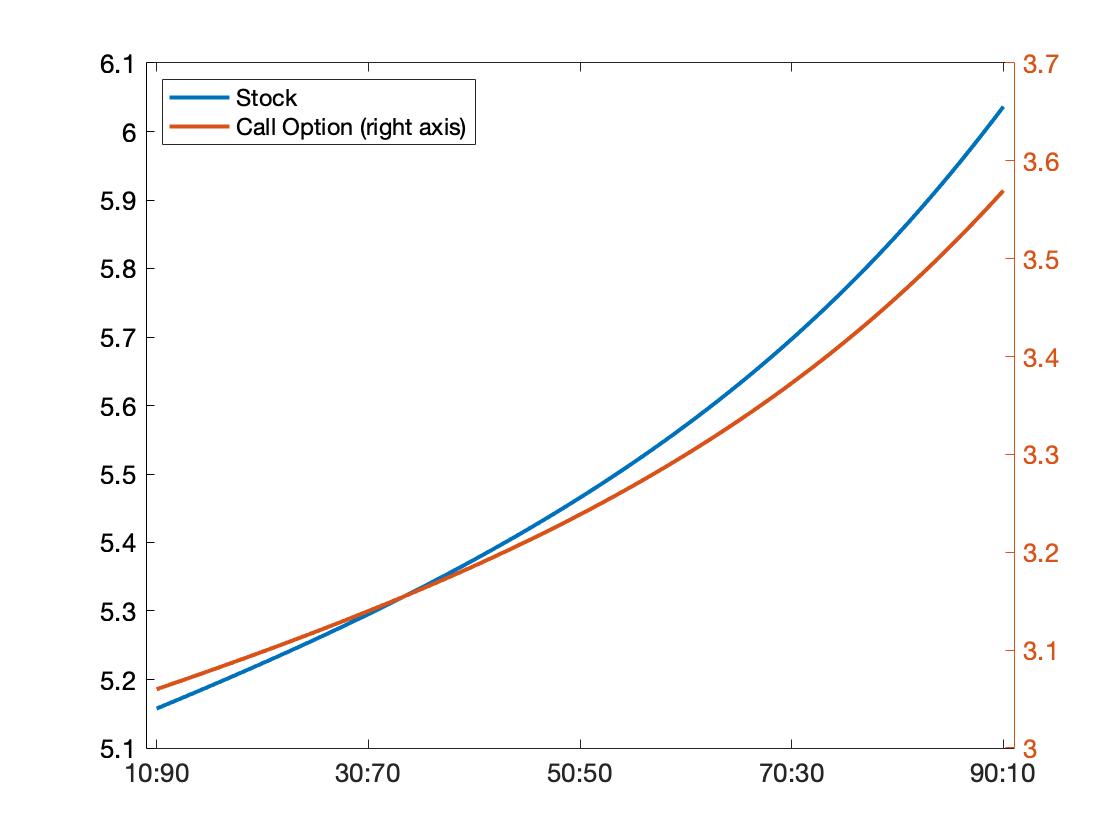}
\end{minipage}
\hfill
\begin{minipage}{0.32\textwidth}
\centering
17\textsuperscript{th} June 2020\\
\includegraphics[width=\textwidth]{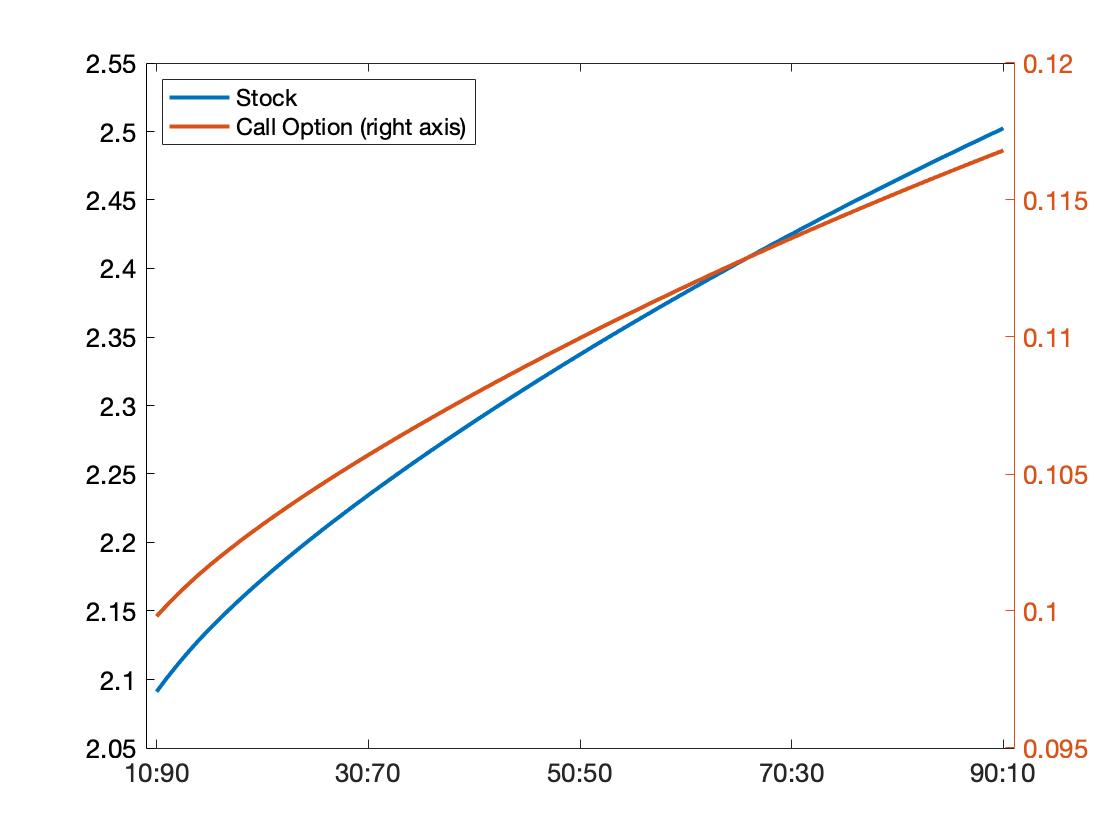}
\end{minipage}
\hfill
\begin{minipage}{0.32\textwidth}
\centering
5\textsuperscript{th} November 2020\\
\includegraphics[width=\textwidth]{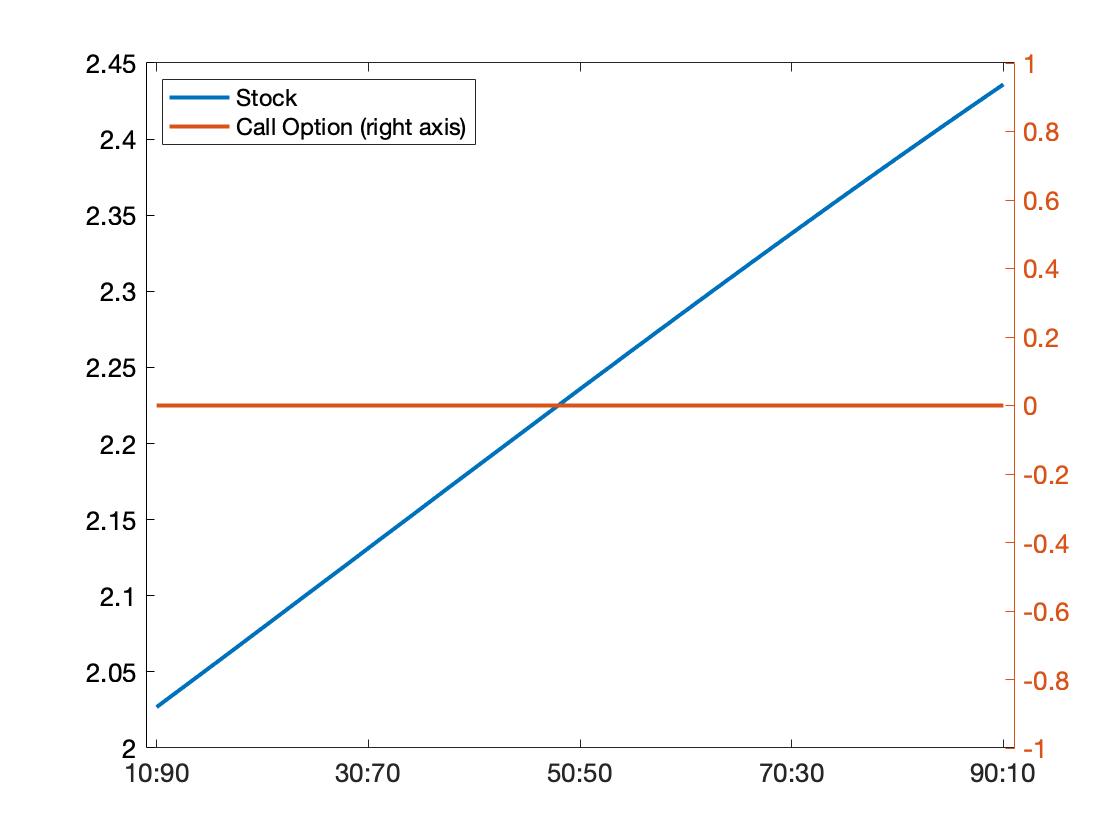}
\end{minipage}
\caption{Lambda quantile risk contributions as of 11\textsuperscript{th} November 2019 (left), 17\textsuperscript{th} June 2020 (centre) and 5\textsuperscript{th} November 2020 (right) of Exxon Mobil Corporation stock and call option (right axis) as a function of portfolio composition with $m=90\%$, $T=2$-years, and issue date of 13\textsuperscript{th} November 2018. On the $x$-axis, 10:90 represents the portfolio composition $(u_s,u_c)=(10,90)$ and 90:10 represents the portfolio composition $(u_s,u_c)=(90,10)$.}
\label{fig:XOM_RC}
\end{figure}

\section{Conclusion}

This paper presents a novel treatment of lambda quantile risk measures on subsets of $\mathbb{R}^n$. We prove that lambda quantiles are differentiable with respect to the portfolio composition, subject to smoothness conditions, and derive explicit formulae for the derivatives. These partial derivatives correspond to risk contributions of assets to the overall portfolio risk. We further provide the Euler decomposition of lambda quantiles, i.e. the property that lambda quantiles, scaled by a factor, can be written as sums of their partial derivatives scaled by the number of assets. This decomposition demonstrates that lambda quantiles are homogeneous risk measures in the space of portfolio compositions. Our results further show that the homogeneity degree of a lambda quantile is determined by the portfolio composition and the lambda function. This allows us to treat homogeneity as a dynamic property rather than constant and universal. Indeed, the lambda quantile homogeneity degree varies across portfolio risk profiles, rather than remaining constant. This contrasts the case of risk measures that have a constant homogeneity degree, such as VaR. Due to the variable nature of lambda quantiles' homogeneity degrees, we introduce a generalised Euler capital allocation rule, that is compatible with risk measures of any homogeneity degree and non-linear portfolios. We prove that the generalised Euler allocations of lambda quantiles have the full allocation property.
On a financial market portfolio, we illustrate and provide interpretation of the homogeneity degree of a non-linear portfolio and its generalised Euler allocation.

\newpage

\appendices
\section{Auxiliary definitions and results}
\label{app:results}
This appendix is a collection of results and definition relevant for the exposition of the paper.

\begin{lemma}[Theorem A.5.1 of \cite{d19}]\label{lem:d19}
Let $(S,\mathcal{S},\mu)$ be a measure space. Let $f$ be a complex valued function defined on $\mathbb{R}\times S$. Let $\delta>0$, and suppose that for $x\in(y-\delta,y+\delta)$ we have:
\begin{enumerate}[(i)]
\item $u(x)=\int_S f(x,s)\mu(ds)$ with $\int_S |f(x,s)|\mu(ds)<\infty$,
\item for fixed $s$, $\partial f/\partial x (x,s)$ exists and is a continuous function of $x$,
\item $v(x)=\int_S \frac{\partial f}{\partial x}(x,s)\mu(ds)$ is continuous at $x=y$,
\item $\int_S \int_{-\delta}^{\delta}|\frac{\partial f}{\partial x}(y+\theta,s)|d\theta\mu(ds)<\infty$,
\end{enumerate}
then $u'(y)=v(y)$.
\end{lemma}

\begin{definition}[Definition 4.2 of \cite{t99}]\label{defn:pm}
Let $U\neq\emptyset$ be a set in $\mathbb{R}^n$ and $r:U\rightarrow\mathbb{R}$ be a function defined on $U$. A vector field $\bm{a}:=(a_1,\dots,a_n):U\rightarrow\mathbb{R}^n$ is called \emph{suitable for performance measurement} with the function $r$ if the following conditions are satisfied:
\begin{enumerate}[(a)]
\item for all $\bm{m}\in\mathbb{R}^n$ and $\bm{u}\in U$ with $r(\bm{u})\neq\bm{m'u}$ and $i\in\{1,\dots,n\}$ the inequality
\begin{equation}\label{ineq1}
m_ir(\bm{u})>a_i(\bm{u})\bm{m'u}
\end{equation}
implies that there exists an $\epsilon>0$ such that for all $t\in(0,\epsilon)$ we have
\begin{equation*}
g_{r,\bm{m}}(\bm{u}-t\bm{e}_i)<g_{r,\bm{m}}(\bm{u})<g_{r,\bm{m}}(\bm{u}+t\bm{e}_i).
\end{equation*}
\item for all $\bm{m}\in\mathbb{R}^n$ and $\bm{u}\in U$ with $r(\bm{u})\neq\bm{m'u}$ and $i\in\{1,\dots,n\}$ the inequality
\begin{equation}\label{ineq2}
m_ir(\bm{u})<a_i(\bm{u})\bm{m'u}
\end{equation}
implies that there exists an $\epsilon>0$ such that for all $t\in(0,\epsilon)$ we have
\begin{equation*}
g_{r,\bm{m}}(\bm{u}-t\bm{e}_i)>g_{r,\bm{m}}(\bm{u})>g_{r,\bm{m}}(\bm{u}+t\bm{e}_i),
\end{equation*}
\end{enumerate}
where $g=g_{r,\bm{m}}:\{\bm{u}\in U|r(\bm{u})\neq\bm{m'u}\}\rightarrow\mathbb{R}$ is the \emph{profit and loss function} for $r$ for a fixed $\bm{m}\in\mathbb{R}^n$ defined by
\begin{equation*}
g_{r,\bm{m}}(\bm{u}):=\frac{\bm{m'u}}{r(\bm{u})-\bm{m'u}}\,.
\end{equation*}
\end{definition}

\begin{lemma}[Theorem 4.4 of \cite{t99}]\label{lem:tspm}
Let $\emptyset\neq U\subset\mathbb{R}^n$ be an open set and $r:U\rightarrow\mathbb{R}$ a function partially differentiable in $U$ with continuous derivatives. Also, let $\bm{a}=(a_i,\dots,a_n):U\rightarrow\mathbb{R}^n$ be a continuous vector field. Then $\bm{a}$ is suitable for performance measurement with function $r$ if, and only if:
\begin{equation*}
a_i(\bm{u})=\frac{\partial r}{\partial u_i}(\bm{u})
\end{equation*}
with $i=1,\dots,n$ and $\bm{u}\in U$.
\end{lemma}

\section*{Funding}

Silvana Pesenti would like to acknowledge the support of the Natural Sciences and Engineering Research Council of Canada (NSERC) with funding reference numbers DGECR-2020-00333 and RGPIN-2020-04289.

\bibliography{ref_rclq}

\end{document}